\documentclass{article}

\usepackage[utf8]{inputenc}
\usepackage{amsmath}
\usepackage{amsfonts}
\usepackage{dsfont}
\usepackage{amsthm,enumitem}
\usepackage[left=2cm,right=2cm,top=3cm,bottom=3cm]{geometry}
\usepackage{graphicx}
\usepackage{hyperref}
\usepackage{caption}
\usepackage{cleveref}
\usepackage[colorinlistoftodos,bordercolor=orange,backgroundcolor=orange!20,linecolor=orange,textsize=scriptsize]{todonotes}
\usepackage{xcolor}
\hypersetup{
    colorlinks,
    linkcolor={red!50!black},
    citecolor={blue!50!black},
    urlcolor={blue!80!black}
}

\usepackage[titletoc,title]{appendix}

\usepackage{algorithm,algcompatible}
\usepackage{algpseudocode}

\newcommand{\R}{\mathbb{R}}
\newcommand{\V}{\mathrm{Var}}
\newcommand{\Q}{\mathbb{Q}}
\newcommand{\Pb}{\mathbb{P}}

\newcommand{\N}{\mathbb{N}}
\newcommand{\E}{\mathbb{E}}
\newcommand{\F}{\mathcal{F}}

\author{Aur\'elien Alfonsi\thanks{CERMICS, Ecole des Ponts, Marne-la-Vall\'ee, France. MathRisk, Inria, Paris, France. E-mail: {\tt aurelien.alfonsi@enpc.fr}}\and  Adel Cherchali\thanks{CERMICS, Ecole des Ponts, Marne-la-Vall\'ee, France. MathRisk, Inria, Paris, France. E-mail: {\tt adel.cherchali@enpc.fr}}\and   Jose Arturo Infante Acevedo\thanks{ GIE AXA, 25 avenue Matignon 75008 Paris, France {\tt josearturo.infanteacevedo@axa.com}. } }
\title{Multilevel Monte-Carlo for computing the SCR with the standard formula and other stress tests}

\theoremstyle{plain}%
\newtheorem{theorem}{Theorem}
\newtheorem{lemma}[theorem]{Lemma}
\newtheorem{proposition}[theorem]{Proposition}

\newtheorem{remark}{Remark}

\theoremstyle{definition}

\newtheorem*{notations}{Notation}

\begin{document}

\maketitle
\setcounter{tocdepth}{3}
\begin{abstract} This paper studies the multilevel Monte-Carlo estimator for the expectation of a maximum of conditional expectations. This problem arises naturally when considering many stress tests and appears in the calculation of the interest rate module of the standard formula for the SCR. We obtain theoretical convergence results that complements the recent work of Giles and Goda~\cite{GiGo} and gives some additional tractability through a parameter that somehow describes regularity properties around the maximum. We then apply the MLMC estimator to the calculation of the SCR at future dates with the standard formula for an ALM savings business on life insurance. We compare it with estimators obtained with Least Squares Monte-Carlo or Neural Networks. We find that the MLMC estimator is computationally more efficient and has the main advantage to avoid regression issues, which is particularly significant in the context of projection of a balance sheet by an insurer due to the path
  dependency. Last, we discuss the potentiality of this numerical method and analyse in particular the effect of the portfolio allocation on the SCR at future~dates. 
\end{abstract}
\section{Introduction}

Solvency II is a regulatory framework introduced in Europe in the period post-financial crisis of 2008. Solvency~II establishes the requirements to be met to exercise the insurance or reinsurance activity in Europe and aims to protect policyholders and to give stability in the financial sector of the European Union. 

One of the advantages of the Solvency~II directive is that the computation required to evaluate the Solvency Required Capital (SCR) considers the specific risks borne by the insurers in comparison to the previous rules where the need of own funds ignored, for example, part of risks embedded in the asset side of the balance sheet. In practice, there are two possible ways to calculate the SCR: the insurance company can either use the standard formula (see~\cite{EIOPA}) by applying shocks to each asset class or use an internal model to calculate a Value-at-Risk over one year.

Today, the SCR indicator is one of the most important Key Performance Indicator used by companies to monitor the activity. In particular, the so-called Solvency~II ratio computed as the ratio between the ``Eligible Own Fund'' and the SCR measures the solvency capacity of the insurers and it is followed by analysts to evaluate them in financial markets.  Nevertheless, it is important to remark that the SCR corresponds to the amount of required capital in a 1 year horizon. Then, to have an idea of the total amount of required capital during the life of a product or over the duration of the business, it is not only necessary to compute the current SCR value but also to estimate the SCR at future dates. 

The aim of this work is to deal with the problem of computing SCR at future dates which has several practical applications to real problems that arise in the insurance industry. One of the first applications that should be cited comes from the regulatory side and is called ORSA (Own Risk and Solvency Assessment) process, which aims to evaluate from a prospective point of view the overall solvency needs related to the specific risk profile of the insurance companies. In order to do that, the computation of future SCR is necessary to ensure that the insurer is able to integrate the regulatory constraints in terms of solvency during the strategic plan horizon. 

Other important applications appear when the notion of cost of capital is concerned. M\"ohr~\cite{Moehr} and Engsner et al.~\cite{ELL} have pointed the importance of considering in this context the SCR values at future dates as random variables and develop a theoretical framework to relate them with the risk margin and the cost of capital. Among the applications related to the cost of capital, one can mention:
\begin{enumerate}
  \item[(i)]	Applications for ALM (Asset Liability Management) when a Strategic Asset Allocation needs to be computed for a given portfolio: to evaluate the optimality of an asset allocation, a criterion based on the sum of the present values of shareholder margins minus the amount of cost of capital generated by the asset allocation is usually studied. The idea of this approach is to analyse if the future gains generated by the portfolio meet the shareholder's expectations in terms of cost of capital.
\item[(ii)]	Applications for pricing, when evaluating if future margins pay the return expected by the shareholders. Before launching a new product, the insurers evaluate the profitability of that product and then compare expected future shareholder margins with the need of capital generated by the new business. 
\end{enumerate}
Finally, the computation of future SCR can be used as a tool for studying the solvency of the company under different economic scenarios. For example, the current low interest rates environment leads to several questions about the solvency of insurance companies and the future sustainability of the Savings business. In addition, the SCR computation at future dates allows to better understand the pattern of cash-flows generated by a product during the lifetime of the business. In particular, the approach based on shocks employed in this work as, for example, the shocks on the market conditions is useful to study the evolution of the balance sheet and the policyholder behavior under those shocked conditions.

Today, the computation of future SCR values are in practice often made with rough extrapolations from the initial SCR, which may ignore the evolution of the risk profile of the insurer and then lead to bad decisions impacting the business. Thus, the goal of the present work is to develop numerical methods for the calculation of the SCR required at some future date~$t$. We focus here on the calculation of SCR with the standard formula, which is fully described by the documents of the European Insurance and Occupational Pensions Authority (EIOPA)~\cite{EIOPA,EIOPA2}. Basically, this standard formula consists in applying different shocks on the different market sectors: the impact on the portfolio of each shock is evaluated in a risk-neutral world, and the SCR is then evaluated by using an aggregation formula from these impacts.
Let us explain more precisely how it works. We consider an insurance company that handles a Savings portfolio up to time~$T>0$ and models the financial market by a vector valued stochastic process $(\xi_t)_{t\ge 0}$ for its risk-neutral valuation. Let us assume for sake of simplicity that this is a diffusion model:
$$\xi_s=\xi_0+\int_0^sb(r,\xi_r)dr +\int_0^s\sigma(r,\xi_r)dW_r, \text{ for }  s\ge 0,$$
where $W$ is a Brownian motion. We assume that the sum of the discounted future profits and losses between times $t\in(0,T)$ and $T$ is given by a function $F_{t,T}(\xi_s,s\in[0,T])$. Note that this function may depend on $(\xi_s,s\in[0,t])$ since the portfolio management may be path-dependent. Let us suppose first that the company calculates~$SCR_0$ (the SCR at the initial time~0) with the standard formula. It has to implement $P$ different shocks that correspond to the different source of risks. These shocks happen just after the initial portfolio allocation, which requires a recalibration of the model by the company. Thus, each shock $i\in\{1,\dots,P\}$ corresponds to a stress on the model parameters, and the market path corresponding to the shock~$i$ is described by:
$$\xi^i_0=\xi_0 \text{ and }\xi^i_s=\check{\xi}^i_0+\int_0^sb^i(r,\xi_r)dr +\int_0^s\sigma^i(r,\xi_r)dW_r, \text{ for }  s> 0. $$
For example, the shock on equity is simply a drop of the initial value  ($\check{\xi}^i_0 \not = \xi_0$) while the shocks on the interest rates require also a recalibration of the model to the shocked yield curve (and thus $b^i$ and $\sigma^i$ may then be different from $b$ and $\sigma$). The P\&L contribution of each shock is given by $$\zeta^i_0=\E[F_{0,T}(\xi_s,s\in[0,T])]-\E[F_{0,T}(\xi^i_s,s\in[0,T])]=\E[F_{0,T}(\xi_s,s\in[0,T])-F_{0,T}(\xi^i_s,s\in[0,T])].$$ The SCR on market risk is finally obtained by using a so-called aggregation formula $SCR_0=h(\zeta^1_0,\dots,\zeta^P_0)$, where $h$ is a function prescribed by the EIOPA. Now, we apply the same methodology at time~$t$ to calculate the required capital $SCR_t$. The market path corresponding to the shock~$i$ at time~$t$ is given by:
$$\xi^i_s=\xi_s \text{ for } s\in[0,t] \text{ and } \xi^i_s=\check{\xi}^i_t+\int_t^sb^i(r,\xi_r)dr +\int_t^s\sigma^i(r,\xi_r)dW_r, \  s> t, $$
and this shock is assessed with the random variable
$$ \zeta^i_t=\E[F_{t,T}(\xi_s,s\in[0,T])-F_{t,T}(\xi^i_s,s\in[0,T])| (\xi_s,s\in[0,t])].$$
It is worth to notice that all the variables $\zeta^i_t$ are defined with the same conditioning: each stress consists in applying a shock at time~$t$, not in considering a particular conditioning set on~$(\xi_s,s\in[0,t])$. Finally, the SCR at time~$t$ is given by $SCR_t=h(\zeta^1_t,\dots,\zeta^P_t)$. For management purposes, it is interesting to assess at the initial time~$0$ how much required capital will be needed in the future. We are thus interested in calculating quantities such as
\begin{equation}\label{PV_SCR}\E[\phi(\xi^i_s,s\in[0,t])SCR_t]=\E[\phi(\xi^i_s,s\in[0,t])h(\zeta^1_t,\dots,\zeta^P_t)],
\end{equation}
where $\phi$ is a function that may include discounting as well as a change of probability between the risk-neutral and the real one.

Many works in the literature deal with the numerical computation of the SCR so that we cannot be exhaustive. Devineau and Loisel~\cite{DeLo}, Bauer et al.~\cite{BaReSi} have investigated numerical methods based on nested simulations.  Bauer et al.~\cite{BaReSi2},  Krah et al.~\cite{KrNiKo} and Floryszczak et al.~\cite{FlLCMa} have used Least Squares Monte-Carlo (regress now) methods for the risk while Pelsser and Schweizer~\cite{PeSc}, Cambou and Filipovi\'c have developed the replicating portfolio (or regress later) approach. Recently, Cheredito et al.~\cite{ChErWu} and Fernandez-Arjona and Filipovi\'c~\cite{FAFi} have proposed to use neural networks to approximate the conditional expectation.  Up to our knowledge, there are however no dedicated studies on the use of multilevel Monte-Carlo estimators for the calculation of the SCR with practical application in an insurance context. This paper fills this gap. Besides, most paper deal with the quantile formulation of the SCR and focus on the calculation of the current value of the SCR (there are few exceptions such as Vedani and Devineau~\cite{VeDe}). Here, we consider instead the calculation of the SCR with the standard formula at future dates. Last, most of the literature either use simple Markovian underlying models or consider instead models from insurance companies that are black boxes, which makes difficult the reproducibility of the results. Here, we are in between and make our experiments on a synthetic ALM model that we recently developed and fully presented in~\cite{AlChIn} which takes into account many path-dependent features of the ALM for life insurance. \\

We now describe the formal mathematical framework and consider a probability space $(\Omega,\F,\Pb)$. Let $X$ and $Y$ be two random variables such that $X$ takes values in a general measurable space $(G,\mathcal{G})$ and $Y$ takes values in $\R^P$, $P\in \N^*$. We make the following assumptions:
\begin{itemize}
\item[(A.1)] $Y$ is square integrable $\R^P$-valued random variable, \label{A1}
\item[(A.2)] $\phi:G\to \R$ is a measurable real-valued function $\phi$ such that $\phi(X)$ is square integrable. \label{A2}
\end{itemize}
For the financial application that we consider in this paper, $X$ represents the market information up to some time~$t>0$. We may thus take $G=\mathcal{C}([0,t],\R^d)$, the space of $\R^d$-valued continuous function, if we consider a market with $d\in \N^*$ continuous assets up to time~$t>0$. We are interested in the problem of computing nested expectations of the form :   
\begin{equation}\label{calc_risq}
I=\E\left[h\left(\E\left[Y^1|X\right],\ldots,\E\left[Y^P|X\right]\right)\phi(X)\right],
\end{equation}
where $h:\R^P \to \R$ is a measurable function with sublinear growth (i.e. $\exists C>0,\forall x \in \R^P , |h(x)|\le C(1+|x|)$), which ensures by Assumptions~(A.1) and~(A.2) that $I$ is well defined. Formula~\eqref{calc_risq} precisely corresponds to~\eqref{PV_SCR} by taking $X=(\xi_s,s\in[0,t])$ and $Y^i=F_{t,T}(\xi_s,s\in[0,T])-F_{t,T}(\xi^i_s,s\in[0,T])$. Thus, $\E[Y^i|X]$ typically represents the expected loss at time~$t$ with shock~$i$, the function $h$ describes the aggregation of the shocks in terms of own funds, and the function $\phi$ weights the different events up to~$t$. 

The calculation of~$I$ is usually made by using a nested Monte-Carlo method: one simulates $J$ independent samples of~$X$ called primary scenarios and then, for each primary scenario, one simulates $K$ independent samples of~$Y$ to approximate the conditional expectations involved in~\eqref{calc_risq} by the corresponding empirical means. This method has been investigated by Gordy and Juneja~\cite{GoJu} and Broadie et. al~\cite{BrDuMo} to calculate the probability of large losses, which amounts to take $P=1$, $h(x)=\mathds{1}_{x>u}$ and $\phi\equiv 1$ in~\eqref{calc_risq} and enables them afterwards to estimate the Value-at-Risk of $\E[Y^1|X]$. Under mild assumptions, they show that the optimal tuning to approximate~$I$ with a precision of $\varepsilon>0$ is to take $J$ proportional to $\varepsilon^{-2}$ and $K$ proportional to $\varepsilon^{-1}$, leading to an overall complexity of $O(\varepsilon^{-3})$. The multilevel Monte-Carlo method (MLMC) developed by Giles~\cite{Giles1} has been applied to the calculation of nested expectations by Haji-Ali~\cite{Haji},  Bujok et al.~\cite{BuHaRe} and Giles~\cite{Giles}. Under some regularity assumptions on~$h$, they show that the antithetic MLMC estimator achieves a precision~$\varepsilon>0$ with a computational cost of~$O(\varepsilon^{-2})$. Under additional regularity assumptions on $h$ or on the probability density function of $(X,Y)$, Giorgi et al.~\cite{GiLePa} have applied the Richardson-Romberg Multilevel method developed by Lemaire and Pagès~\cite{LePa} to improve the convergence of the MLMC estimator.

In this work, we focus on the case where $h$ is the maximum function: 
\begin{equation}\label{def_I_max}
I=\E\left[\max\left\{\E\left[Y^1|X\right],\ldots,\E\left[Y^P|X\right]\right\}\phi(X)\right]. 
\end{equation}
This function is sublinear, but is not differentiable when two (or more) arguments achieve the maximum. Due to this singularity, the general result given by Giles~\cite[Section 9.1]{Giles} when $h$ is twice differentiable does not apply, and a careful analysis is required. Note that~\eqref{def_I_max} appears in the standard formula for the calculation of the SCR interest rate module. To be precise, one has to compute in this case $\E\left[\max\left\{\E\left[Y^1|X\right],\ldots,\E\left[Y^P|X\right],0\right\}\phi(X)\right]$, which amounts to add a zero coordinate to~$Y$. More generally, the problem of computing~\eqref{def_I_max} occurs  when one has to determine the worst of a set of $P$ shocks (or stress tests) on a portfolio of securities at some future time~$t$.  When the function $\phi$ is nonnegative and such that $\E[\phi(X)]=1$, $\phi(X)$ can be seen as a change of probability on the different events up to time~$t$. The function $\phi(X)$ does not add any technical difficulty in our study, but it enables us to perform the evolution up to time~$t$ under the real probability and the evaluation of the losses under the risk-neutral probability, as it is recommended by Solvency II. It can also include some discounting factor. Studies of MLMC estimators for nested expectations for irregular functions~$h$ with applications to risk management have recently been made by Giles and Haji-Ali~\cite{GiHA}, Bourgey et al.~\cite{BDMGZ} and Giorgi et al.~\cite{GiLePa}. In a very recent work, Giles and Goda~\cite{GiGo} have studied precisely the problem of computing~\eqref{def_I_max} with the MLMC method.

The contribution of this paper is twofold. First, we provide an original mathematical analysis of the MLMC estimator for the calculation of~\eqref{def_I_max} that completes the result obtained by Giles and Goda~\cite{GiGo}. Our analysis relies on different arguments and the required assumptions are therefore also different. In particular, Giles and Goda make some technical assumptions to control the probability of two elements being close to the maximum. These assumptions are replaced in our analysis by an integrability assumption involving a parameter $\eta \in (0,1)$ that gives some additional flexibility in the application of the MLMC estimator. Our second contribution is to apply this method to an ALM model for life insurance that takes into account the main characteristic of the business: book values, profit-sharing mechanism, minimum guaranteed rate, etc. Thus, the model is truly path-dependent so that the conditional expectation at time~$t$ really involves the past dynamics, which makes the use of regression techniques more delicate. Indeed, one of the main advantage of the MLMC estimator is to calculate directly~$I$ and skip the regression issue. The second main advantage is that it provides an estimator with accuracy~$\varepsilon$ and with a computational cost in $O(\varepsilon^{-2})$: it is thus asymptotically as efficient as a Monte-Carlo method for plain expectations. In our numerical study, we compare the estimation of~$I$ with MLMC, Least Squares Monte-Carlo (LSMC) estimator and the use of Neural Networks (NN), and demonstrate the main advantages of the MLMC estimator.

The paper is organized as follows. Section~\ref{Sec_Math} presents the mathematical results on the estimation of~$I$ with nested Monte-Carlo and MLMC. Technical proofs are postponed to Appendix~\ref{Appendix_MLMC}. Section~\ref{Sec_ALM} then deals with the application to ALM. Subsections~\ref{subsec_ALM} and~\ref{subsec_MODEL} present the ALM model for life insurance business that we developed in~\cite{AlChIn} while Subsection~\ref{subsec_STD} recalls the calculation of the SCR with the standard formula. Subsection~\ref{subsec_NUMI} compares the numerical performance of the MLMC estimator with estimators obtained with  LSMC or NN. Last, Subsection~\ref{subsec_NUMII} shows the interest of analysing the SCR at future dates, exhibiting some interesting properties such as the dependence of the SCR on the portfolio allocation or on the market risk premia.

\section{Mathematical analysis of Monte-Carlo estimators of~$I$}\label{Sec_Math}

\subsection{Nested Monte-Carlo estimator}\label{Subsec_NMC}

In order to compute $I$ defined by~\eqref{def_I_max}, the classical approach is to approximate the inner and outer expectation by using  Monte-Carlo estimators. The procedure consists in generating an i.i.d sample $(X_1,\ldots,X_J)$ of $X$ called outer (or primary) scenarios. Then, conditionally on $X_i$, we sample $(Y_{i,1},\ldots,Y_{i,K})$  called inner (or secondary) scenarios following the conditional law of $Y$ given $X=X_i$ and approximate the conditional expectation $\E\left[Y^p|X=X_i\right]$ by 
\begin{equation} 
 \widehat{E}^{p}_{i,K}=\frac{1}{K}\sum_{k=1}^K Y^{p}_{i,k},
\end{equation}
for $i\in\{1,\dots,J\}$ and $p\in \{1,\ldots,P\}$. The outer expectation is then approximated using the standard MC estimator:
\begin{equation}
\widehat{I}_{J,K}=\frac{1}{J}\sum_{j=1}^J \max\left\{\widehat{E}^{1}_{j,K},\ldots,\widehat{E}^{P}_{j,K}\right\}\phi(X^{j}) 
\label{Nested_estimator}
\end{equation}
Nested Monte-Carlo estimators have been studied for example by Gordy and Juneja~\cite{GoJu} in the context of portfolio risk measurement. These estimators introduce two levels of error: one for the inner expectation and the other one for the outer expectation. In a standard way, we analyse the Mean-Square Error (MSE) of the estimator $\text{MSE}(\widehat{I}_{J,K})=\E\left[\vert\widehat{I}_{J,K}-I\vert^2\right]$ and use the bias-variance decomposition:
$$ \text{MSE}(\widehat{I}_{J,K})= \text{bias}^2(\widehat{I}_{J,K}) +\V(\widehat{I}_{J,K}),$$
where $\text{bias}(\widehat{I}_{J,K}) =\E[\widehat{I}_{J,K}]-I$.

\begin{notations}
\begin{itemize}
\item We set $E_X^p:=\E[Y^p|X]$ for $p\in \{1,\ldots,P\}$ and $M_X^p=\max\{\E[Y^1|X],\ldots,\E[Y^p|X]\}$.
\item Let $K\in\N^*$ and $Y_1,\dots,Y_K$ be an i.i.d. sample following the conditional law of~$Y$ given~$X$. Then, we set
  \begin{equation}\label{Notation1} \forall p=1,\ldots,P, \quad \widehat{E}^{p}_K=\frac 1K \sum_{k=1}^K Y^p_k \text{ and } \widehat{M}_K^{p}=\max\{\widehat{E}^{1}_K,\ldots,\widehat{E}^{p}_K\}.
  \end{equation}
  
\item Besides, when $K$ is even, we define
   \begin{equation}\label{Notation2} \forall p=1,\ldots,P, \quad \widehat{E}^{p,\prime}_{K/2}=\frac 2K \sum_{k=K/2+1}^K Y^p_k \text{ and } \widehat{M}_{K/2}^{p,\prime}=\max\{\widehat{E}^{1,\prime}_{K/2},\ldots,\widehat{E}^{p,\prime}_{K/2}\}.
   \end{equation}
\end{itemize}
\end{notations}
\noindent From the LLN, we have $\widehat{E}^{p}_K\to E_X^p$ and $\widehat{M}_K^{p}\to M_X^p$ almost surely as $K\to +\infty$. The next theorem analyses the MSE of the nested estimator and provides estimates that will be then useful for the analysis of the MLMC estimator.

\begin{theorem}\label{thm_nested}
  Let $P\ge 2$ and $\eta\in (0,1]$.   
    We assume that (A.1) and (A.2) hold, and we define, for $p\in \{1,\dots,P\}$, $\sigma_p(X)=\sqrt{\V(Y^p|X)}$. Let us set $\Sigma^{1+\eta}_P(X)=\sum_{i=1}^P \sigma^{1+\eta}_i(X)$ 
    and
  $$ C_P(X)=2^\eta \Sigma^{1+\eta}_P(X) \sum_{p=2}^P \frac{1}{|E^{p}_X-M^{p-1}_X|^\eta}.$$
 Assume that the following condition holds:
\begin{itemize}
\item[(i)] $\forall\ p=2,\ldots, P,\ \mathbb{P}\left(M^{p-1}_X=E_X^{p}\right)=0$, 
\item[(ii)] $\Sigma^2= \E[\Sigma^2_P(X) \phi^2(X)]<\infty$ and $C=\E[C_P(X)|\phi(X)|]<\infty$.
\end{itemize}
 Then, we have
 \begin{align}
  \left|  \E \left(\left(\widehat{M}^P_K-M^P_X\right)\phi(X)\right) \right| \leq \frac{C}{K^{\frac{1+\eta}{2} }} \ \text{ and }  \E\left(\left(\widehat{M}^P_K-M^P_X\right)^2\phi^2(X)\right) \leq \frac{\Sigma^2}{K}. \label{BV_nested}
 \end{align}
 Besides, if $V=\V(M^P_X \phi(X))<\infty$, we get
 \begin{equation}\label{MSE_nested} \text{MSE}(\widehat{I}_{J,K})\le \frac{C^2}{K^{1+\eta}}+ \frac{2 V } {J } +\frac{2\Sigma^2} {J K}.  
 \end{equation}
With this upper bound, taking $K\sim c\varepsilon^{-\frac{2}{1+\eta}}$ and $J\sim c'\varepsilon^{-2}$ for some constants $c,c'>0$  is an asymptotically optimal choice to get  $\text{MSE}(\widehat{I}_{J,K})=O(\varepsilon^{2})$ while minimizing the computation cost $JK$.
\label{thm_biais_variance_nested}
\end{theorem}
\begin{remark} Let us note that the assumptions (i) and  $C<\infty$ of Theorem~\ref{thm_biais_variance_nested} are only needed to improve the upper bound on the bias. If it does not hold, we still have
  $$ \left|  \E \left(\left(\widehat{M}^P_K-M^P_X\right)\phi(X)\right) \right| \le  \E \left(\left|\left(\widehat{M}^P_K-M^P_X\right)\phi(X)\right| \right)\le \frac{\Sigma}{\sqrt{K}}, $$
from the right hand side of~\eqref{BV_nested} and Cauchy-Schwarz inequality. Note that this speed of $O(K^{-1/2})$ is the best rate of convergence without further assumption, as illustrated by the following example. Consider the case where $Y=(Y^1,Y^2)$ and, given~$X$, $Y^1$ and $Y^2$ are independent normal distribution with unit variance ($\sigma_1(X)=\sigma_2(X)=1$) and the same mean $m(X)$. Then, $M^2_X=m(X)$ and, given $X$, $\widehat{M}^2_K-M^2_X$ has the same law as $\frac 1 {\sqrt{K}} \max(G^1,G^2)$ where $G^1$ and $G^2$ are independent standard normal variables. Thus, we have for $\phi \equiv 1$: $\E\left[ \left|\widehat{M}^2_K-M^2_X\right| \right]=\frac{c}{\sqrt{K}}$ with $c=\E[|\max(G^1,G^2)|]$.
\end{remark}
\begin{remark} For practical applications such as the standard formula for the SCR interest rate module, one usually considers the positive part of the maximum. This amounts to add the coordinate $Y^{P+1}=0$ in our framework. Thus, if we assume in addition that $\Pb(M^P_X=0)=0$ and $\tilde{C}=\E\left[ \left(C_P(X)+\frac{2^\eta\Sigma^{1+\eta}_P(X)}{|M^P_X|^\eta}\right)|\phi(X)|\right]<\infty$, then 
$$\left|  \E \left(\left((\widehat{M}^P_K)^+-(M^P_X)^+\right)\phi(X)\right) \right| \leq \frac{\tilde{C}}{K^{\frac{1+\eta}2}} \ \text{ and }  \E\left(\left((\widehat{M}^P_K)^+-(M^P_X)^+\right)^2\phi^2(X)\right) \leq \frac{\Sigma^2}{K}.$$
\end{remark}
\begin{remark}
  Let us assume for simplicity that $\phi\equiv1$ and there exists $\underline{\sigma},\overline{\sigma} \in \R_+^*$ such that for all $p\in \{1,\dots,P\}$,
  $$ \underline{\sigma}\le \sigma_p(X) \le \overline{\sigma}, \ a.s.$$
  Then, the integrability condition~(ii) of Theorem~\ref{thm_biais_variance_nested} is equivalent to have $\E[|E_X^p-M_X^{p-1}|^{-\eta}]<\infty$ for all $p\in \{2,\dots,P\}$. Suppose now that $E_X^p-M_X^{p-1}$ admits a probability density~$f_p(x)$ that is continuous and does not vanish at~$0$.  Then, the integrability condition near~$0$ gives $$\int_{-\varepsilon}^{\varepsilon}|x|^{-\eta}f_p(x)dx<\infty \iff \eta<1.$$
This indicates that, in a quite general framework,  condition~(ii) of Theorem~\ref{thm_biais_variance_nested} is not satisfied for $\eta=1$ but may be satisfied for any $0<\eta<1$.
\end{remark}

The proof of Theorem~\ref{thm_biais_variance_nested} is a consequence of the next lemma, whose proof is postponed to Appendix~\ref{subsec_app_nested}. The analysis is rather standard, but the difficulty is to handle in the bias analysis the irregularity of the maximum when two (or more) arguments equal. This is why we need Assumption $(i)$ and the finiteness of~$C$ in Assumption~(ii). These assumptions are different from Assumptions~2 and~3 that are used by Giles and Goda~\cite{GiGo} in a similar context. With their assumptions, they obtain a bias in $O(1/K^{1-\delta}) $ for any arbitrary $0<\delta<1$. Here, we directly see the link between the integrability assumption and the bias in $O(1/K^{\frac{1+\eta}{2}})$. Besides, let us note  that we do not need to assume the boundedness of any moments of $Y^p\phi(X)$, $p\in \{1,\dots,P\}$ (Assumption~1 of~\cite{GiGo}) since we are using a different approach that does not make use of the Burkholder-Davis-Gundy inequality. 
\begin{lemma}
Let  $\eta \in (0,1]$. For $i\in \{1,2\}$, we consider real valued random variables $\widehat{\theta}_K^{i}$ and real valued functions $\varphi_i$ that satisfy the following conditions :
\begin{itemize}
\item[(i)] $\widehat{\theta}_K^{i}\xrightarrow[K \to \infty]{}\varphi_i(X)\ a.s$
\item[(ii)] There are nonnegative measurable functions $C_i$ and $\sigma_i^2$ such that for all $K\in \N^*$:
\begin{equation}
\left\vert \mathbb{E}\left[\widehat{\theta}_K^{i}-\varphi_i(X)|X\right]\right\vert\leq \frac{C_i(X)}{K^{\frac{1+\eta}2}},
\end{equation}
\begin{equation}
 \mathbb{E}\left[\left\vert\widehat{\theta}_K^{i}-\varphi_i(X)\right\vert^2|X\right]\leq \frac{\sigma_i^2(X)}{K}.
\end{equation}
\item[(iii)] $\mathbb{P}\left(\vert\varphi_{12}(X)\vert=0\right)=0$, where $\varphi_{21}(x):=\varphi_2(x)-\varphi_1(x)$.
\end{itemize}
Then, we have with
\begin{align}
C(X)&=\mathds{1}_{\varphi_{21}(X)<0}C_1(X)+\mathds{1}_{\varphi_{21}(X)>0}C_2(X)+2^\eta \frac{\sigma_1^{1+\eta}(X)+\sigma_2^{1+\eta}(X)}{\vert \varphi_{21}(X)\vert},\\
\sigma^2(X)&=\sigma_1^2(X)+\sigma_2^2(X),
\end{align}
the following estimates:
\begin{equation}
\left\vert \mathbb{E}\left[\max\{\widehat{\theta}_K^{1},\widehat{\theta}_K^{2}\}-\max\{\varphi_1(X),\varphi_2(X)\}|X\right]\right\vert\leq \frac{C(X)}{K^{\frac{1+\eta}2}},\label{weak_err_max}
\end{equation}
\begin{equation}
 \mathbb{E}\left[\left\vert\max\{\widehat{\theta}_K^{1},\widehat{\theta}_K^{2}\}-\max\{\varphi_1(X),\varphi_2(X)\}\right\vert^2|X\right]\leq \frac{\sigma^2(X)}{K}.
 \label{strong_err_max}
\end{equation}
\label{lemma_rec_max}
\end{lemma}

\begin{proof}[Proof of Theorem~\ref{thm_biais_variance_nested}]
  We first prove by induction on $P \ge 2$ that
  \begin{align}\left|  \E \left(\widehat{M}^P_K-M^P_X \bigg| X\right) \right| \leq \frac{C_P(X)}{K^{\frac{1+\eta}2}} \ \text{ and }  \E\left(\left(\widehat{M}^P_K-M^P_X\right)^2\bigg|X\right) \leq \frac{\Sigma^2_P(X)}{K}.\label{induc_hyp}
  \end{align}
  We apply Lemma~\ref{lemma_rec_max}, noticing that $\E[\widehat{E}^p_K-E^p_X|X]=0$ and $\E[(\widehat{E}^p_K-E^p_X)^2|X]=\frac{\sigma^2_p(X)}{K}$, for $p\in \{1,\dots,P\}$. First, this gives the result for $P=2$. Second, with the induction hypothesis for $P$, Lemma~\ref{lemma_rec_max} gives that~\eqref{induc_hyp} is satisfied for $P+1$ with
  $$ C_{P+1}(X)=C_P(X)+2^\eta\frac{\Sigma_P^{1+\eta}(X)+\sigma_{P+1}^{1+\eta}(X)}{|\widehat{E}^{P+1}_K-M^P_X|^\eta} \text{ and }  \Sigma^2_{P+1}(X)=\Sigma^2_P(X)+\sigma_{P+1}^2(X),$$
  which gives the claim.

  Since $\text{bias}(\widehat{I}_{J,K})=\E\left[\left(\widehat{M}^P_K-M^P_X\right) \phi(X)\right]$, we get $|\text{bias}(\widehat{I}_{J,K})|\le \frac{\E[C_P(X)|\phi(X)|]}{K^{\frac{1+\eta}2}}=\frac C {K^{\frac{1+\eta}2}}$. Similarly, we have
  \begin{align*}
    \V(\widehat{I}_{J,K})=\frac1J \V[\widehat{M}^P_K\phi(X)] &\le \frac 2J \V\left[\left(\widehat{M}^P_K-M^P_X\right) \phi(X)\right]+\frac 2J \V[M^P_X \phi(X)] \\
    & \le \frac 2J \E\left[\left(\widehat{M}^P_K-M^P_X\right)^2 \phi^2(X)\right]+\frac 2J \V[M^P_X \phi(X)],
  \end{align*}
  which leads to~\eqref{MSE_nested}.

  Last, we notice that for $c_1,c_2>0$, the minimization of $JK$ given $\frac{c_1}{K^{1+\eta}}+\frac{c_2}J=\varepsilon^2$ leads to $J= \frac{c_2}{(1+\eta)c_1} K^{1+\eta}$ and thus $K\sim c\varepsilon^{-\frac{2}{1+\eta}}$ and $J\sim c'\varepsilon^{-2}$ for some $c,c'>0$. Since $\frac1{JK}\le \frac 1J$ and $\frac1{JK}=O(\varepsilon^{2+\frac{2}{1+\eta}})$ is negligible with respect to $\frac 1 {K^2}$ and $\frac 1J$, this choice is asymptotically optimal: it gives $MSE(\widehat{I}_{J,K})=O(\varepsilon^2)$ with a computational cost in $O(\varepsilon^{-3-\frac{1-\eta}{1+\eta}})$. 
\end{proof}

\subsection{The Multilevel Monte-Carlo estimator}

The Multilevel Monte-Carlo (MLMC) is a general method to reduce the computational complexity of estimating an expected value with a biased estimator. Under suitable assumptions, it can even lead to the same asymptotic computational cost as an unbiased Monte-Carlo estimator.
It has originally been developed by Giles~\cite{Giles1} for Stochastic Differential Equations, when the solution is approximated by a discretization scheme such as the Euler scheme. Its application to the calculation of nested expectations has been developed by Haji-Ali~\cite{Haji},  Bujok et al.~\cite{BuHaRe} and Giles~\cite{Giles}. For a detailed presentation of the method, from its origins to its applications, we refer to Giles~\cite{Giles}. Let us explain its principle in few words and suppose that we are interested in computing $\E[\xi]$ with a family of estimators $\xi_l$ such that  $\E[\xi_l] \underset{l\to \infty}{\to} \E[\xi]$. We assume that  $\xi_l$ can be simulated and the larger is~$l$, the more the simulation of~$\xi_l$ requires computation time. Then, the basic idea is to take a large value of $L$ and to consider the Monte-Carlo estimator $\frac 1 {J} \sum_{j=1}^{J} \xi_{L,j}$. However, in many situations, it is possible to simulate jointly $(\xi_{l-1},\xi_l)$ in such a way that the variance of $\Delta_l=\xi_l-\xi_{l-1}$ is much smaller than the variance of $\xi_{l-1}$ and $\xi_l$. Then, observing that $\E[\xi_L]=\E[\xi_0]+\sum_{l=1}^L\E[\Delta_l]$, we consider the following MLMC estimator
\begin{equation}\label{MLMC_gen}
  \frac 1 {J_0} \sum_{j=1}^{J_0} \xi_{0,j} + \sum_{l=1}^L \frac 1 {J_l} \sum_{j=1}^{J_l}\Delta_{l,j},
\end{equation}
where the variables $\xi_{0,j},\Delta_{1,j}\dots,\Delta_{L,j}$ are sampled independently. This estimator has the same bias as $\frac 1 {J} \sum_{j=1}^{J} \xi_{L,j}$ but may have, for a given computational cost, a much lower variance. Since the computational cost of the simulation of $\xi_{0,j}$ is much lower than the one of~$\xi_{L,j}$, one may use a large number of simulations for the level 0 ($J_0>>J$) to reduce the statistical error. For the other levels $l\in \{1,\dots,L\}$, it is instead possible to use a relatively small number of simulations~$J_l$ thanks to the variance reduction given by the joint simulation of $(\xi_{l-1},\xi_l)$. The optimal tuning of $L$ and $J_0,\dots,J_L$ clearly depends on the context, and is analysed for a general framework in~\cite[Theorem 1]{Giles}. Here, we apply this method when $\xi_l$ is the nested Monte-Carlo estimator studied in Subsection~\ref{Subsec_NMC}. In this case, the bias comes from the approximation of the inner expectation by a second Monte-Carlo estimator.

We now present more precisely the MLMC estimator of~$I$ defined by~\eqref{def_I_max}. We consider $L\in \N^*$ that represents the number of levels and    $J_0,\dots,J_L\in \N^*$ that are the number of primary scenarios for each level. We consider $K_0,\dots,K_L \in \N^*$ that describe the numbers of inner simulations and are such that
\begin{equation}
  \forall l \in \{1,\dots,L\},\ K_l=K_02^l. \label{def_Kl}
\end{equation}
For each level $l\in \{0,\dots, L\}$, we consider $(X_{l,j}, 1\le j\le J_l)$ i.i.d.~random variables having the same distribution as~$X$, and  random variables $(Y_{l,j,k}, 1\le j\le J_l, 1\le k\le K_l)$  that are independent given $(X_{l,j}, 1\le j\le J_l)$ and such that $Y_{l,j,k}$ follows the distribution of~$Y$ given $X=X_{l,j}$. These random variables are assumed to be independent between levels, i.e. $(X_{l,j}, Y_{l,j,k}, 1\le j\le J_l, 1\le k\le K_l)_{l\in 0,\dots,L}$  are independent. Then, we define for $l\in \{0,\dots,L\}$ and $p\in \{1,\dots,P\}$:
\begin{align}
  \widehat{E}^p_{l,j,K}&=\frac{1}{K}\sum_{k=1}^{K} Y^p_{l,j,k},\ K\in \{1,\dots,K_l\}\\
  \widehat{M}^p_{l,j,K}&=\max(\widehat{E}^1_{l,j,K},\dots,  \widehat{E}^p_{l,j,K})
\end{align}
Then, the MLMC estimator of~$I$ is defined by
\begin{equation}
  \widehat{I}^{MLMC}=\frac{1}{J_0}\sum_{j=1}^{J_0}\widehat{M}^P_{0,j,K_0}\phi(X_{0,j})+\sum_{l=1}^L \frac{1}{J_l}\sum_{j=1}^{J_l}( \widehat{M}^P_{l,j,K_l} -\widehat{M}^P_{l,j,K_{l-1}})\phi(X_{l,j}). \label{Standard_Nested_MLMC_estim}
\end{equation}
Note that the sum $\frac{1}{J_0}\sum_{j=1}^{J_0}\widehat{M}^P_{0,j,K_0}\phi(X_{0,j})$ of level~$0$ is a nested Monte-Carlo with $J_0$ outer simulations and $K_0$ inner simulations, while the sum  $\frac{1}{J_l}\sum_{j=1}^{J_l}( \widehat{M}^P_{l,j,K_l} -\widehat{M}^P_{l,j,K_{l-1}})\phi(X_{l,j})$ is the difference of two nested Monte-Carlo estimators with $J_l$ outer simulations with $K_l$ and $K_{l-1}$ inner simulations that are computed with the same random variables. We recall in this context the heuristical explanation of the MLMC estimator~\eqref{Standard_Nested_MLMC_estim}. The level~$0$ uses a large number of outer simulations with a small number of inner simulation: this reduces the statistical error but leaves some bias. The other levels $l\in\{1,\dots,L\}$ are bias corrections that have a much smaller statistical error:  $\widehat{M}^P_{l,j,K_l}$ and $\widehat{M}^P_{l,j,K_{l-1}}$ are close to each other since they are constructed from the same random variables. The MLMC estimator uses this idea and we have to determine the values of $J_l$ and $L$ that are asymptotically optimal to estimate~$I$ with a precision~$\varepsilon>0$.

Let us assume that the assumptions of Theorem~\ref{thm_biais_variance_nested} hold. We have $\text{bias}(\widehat{I}^{MLMC})=\E\left[\left(\widehat{M}^P_{K_L}-M^P_X\right) \phi(X)\right]=O(K_L^{-\frac{1+\eta}2})=O(2^{-\frac{1+\eta}2 L})$. Besides, we have
$$\V\left(\left(\widehat{M}^P_{K_l}-\widehat{M}^P_{K_{l-1}}\right) \phi(X) \right)\le 2 \V\left(\left(\widehat{M}^P_{K_l}-M^P_X\right) \phi(X) \right)+ 2 \V\left(\left(\widehat{M}^P_{K_{l-1}}-M^P_X\right) \phi(X) \right)=O(K_l^{-1})=O(2^{-l})$$ and the computational cost of $( \widehat{M}^P_{l,j,K_l} -\widehat{M}^P_{l,j,K_{l-1}})\phi(X_{l,j})$ is $O(K_l)=O(2^l)$. We can thus apply Theorem~1~\cite{Giles}, which leads to the following result.
\begin{proposition}
  Let us assume that the assumptions of Theorem~\ref{thm_biais_variance_nested} hold for some $\eta \in(0,1]$. Then, by taking when $\varepsilon \to 0$
  \begin{equation}\label{reglage_multi}
    L=\left\lceil  \frac{2}{1+\eta} \frac{|\log(\varepsilon)|}{\log(2)}\right\rceil,  J_0=2^{\left\lceil\frac{ 2|\log(\varepsilon)|+|\log(|\log(\varepsilon)|)|}{\log(2)}  \right\rceil}=O(\varepsilon^{-2}|\log(\varepsilon)|) \text{ and } J_l=J_02^{-l}, l\in\{1,\dots L\},
  \end{equation}
  we have $MSE(\widehat{I}^{MLMC})=\E[(\widehat{I}^{MLMC}-I)^2]=O(\varepsilon^2)$ with a computational cost in $O(\varepsilon^{-2}\log^2(\varepsilon))$.

  If only the assumption $\Sigma^2<\infty$ of Theorem~\ref{thm_biais_variance_nested} holds, the same conclusion holds by taking~$\eta=0$ in~\eqref{reglage_multi}.
\end{proposition}
\begin{proof}
  We just check that the parameters achieve the claim. From the bias-variance decomposition, we get by using Theorem~\ref{thm_nested},~\eqref{def_Kl} and~\eqref{reglage_multi} that there is a positive constant $C$ such that
  $$MSE(\widehat{I}^{MLMC})\le C\left( \frac 1 {K_L^{1+\eta}} + \frac 1 {J_0} + \sum_{l=1}^L \frac 1 {J_lK_l} \right)= C\left( \frac{2^{-(1+\eta) L}}{K_0^{1+\eta}} + \frac 1 J_0 +  \frac L {J_0K_0} \right). $$
  The choice of $L$ gives  $2^{-(1+\eta) L}\le \varepsilon^2$ and the choice of $J_0$ then gives $\frac L {J_0} =O(\varepsilon^2)$. Last the computational cost is given by $\sum_{l=0}^LJ_lK_l=LJ_0K_0=O(\varepsilon^{-2}\log^2(\varepsilon))$.
  In the case where we only know $\Sigma^2<\infty$, only the second statement of Equation~\eqref{BV_nested} holds, and we get
  $$\left|\E\left[\left(\widehat{M}^P_{K_L}-M^P_X\right) \phi(X)\right]\right|\le \frac{\Sigma}{\sqrt{K_L}}=\frac{\Sigma}{\sqrt{K_0}}2^{-L/2},$$
which gives the second claim with the same arguments. 
\end{proof}
\begin{remark}
  Let us note that the analysis of the computational cost gives that it is asymptotically bounded by $C\varepsilon^2\log^2(\varepsilon)$ for some constant $C>0$, but it does not analyse precisely this constant. Nonetheless, since this cost is $LJ_0K_0$, this constant can be chosen to be proportional to the number of levels. 

Thus, the analysis of the bias given by Theorem~\ref{thm_biais_variance_nested} under the integrability assumption $\E[C_P(X)|\phi(X)|]<\infty$ enables to reduce the number of levels and then to reduce this constant. 
\end{remark}

It is however possible to construct a better estimator by using the following MLMC antithetic estimator
\begin{equation}
  \widehat{I}^{MLMC}_A=\frac{1}{J_0}\sum_{j=1}^{J_0}\widehat{M}^P_{0,j,K_0}\phi(X_{0,j})+\sum_{l=1}^L \frac{1}{J_l}\sum_{j=1}^{J_l} \left( \widehat{M}^P_{l,j,K_l} -\frac{\widehat{M}^P_{l,j,K_{l-1}}+\widehat{M}^{P,\prime}_{l,j,K_{l-1}}}{2}\right)\phi(X_{l,j}), \label{Antithetic_MLMC_estim}
\end{equation}
where we set for $p\in \{1,\dots,P\}$,
\begin{equation}\label{def_prime}
   \widehat{E}^{p,\prime}_{l,j,K_{l-1}}=\frac{1}{K_{l-1}}\sum_{k=K_{l-1}+1}^{K_l} Y^p_{l,j,k} \text{ and }  \widehat{M}^{p,\prime}_{l,j,K_{l-1}}=\max(\widehat{E}^{1,\prime}_{l,j,K_{l-1}},\dots,  \widehat{E}^{p,\prime}_{l,j,K_{l-1}}).
\end{equation}
This is a rather natural idea to reduce the variance contribution of each level, see Section~9.1 of~\cite{Giles}. However, the irregularity of the maximum function makes the analysis of the variance more delicate as if it were a smooth function. Giles and Goda~\cite{GiGo} give an analysis of the variance that require the boundedness of any moments of $Y^p\phi(X)$, $p\in \{1,\dots,P\}$ (Assumption~1 of~\cite{GiGo}) and assumptions to control the probability that another component is close to the maximum (Assumptions~2 and~3 of~\cite{GiGo}). Here, our proof relies on a different argument. It only requires a moment condition that quantifies in a different way the probability that two or more arguments in the maximum are close to the maximum. The parameter $\eta \in (0,1]$ involved in this condition gives besides more flexibility to apply the MLMC method. Details are in the appendix (see Proposition~\ref{prop_antithetic}).  

\begin{remark}\label{Rk_Anti_gen}
  For the calculation of~\eqref{calc_risq} with a general function~$h$, the antithetic MLMC estimator is defined by
  \begin{align*}
    &\frac{1}{J_0}\sum_{j=1}^{J_0}h(\widehat{E}^{1}_{0,j,K_{0}},\dots,\widehat{E}^{P}_{0,j,K_{0}})\phi(X_{0,j})\\
    &+\sum_{l=1}^L \frac{1}{J_l}\sum_{j=1}^{J_l} \left( h(\widehat{E}^{1}_{l,j,K_{l}},\dots,\widehat{E}^{P}_{l,j,K_{l}})  -\frac{h(\widehat{E}^{1}_{l,j,K_{l-1}},\dots,\widehat{E}^{P}_{l,j,K_{l-1}}) + h(\widehat{E}^{1,\prime}_{l,j,K_{l-1}},\dots,\widehat{E}^{P,\prime}_{l,j,K_{l-1}})}{2}\right)\phi(X_{l,j}) .
  \end{align*}
  In particular, it is possible to estimate by MLMC the value of~\eqref{calc_risq} for different functions~$h$ with the same simulations. 
\end{remark}

\begin{theorem}\label{thm_MLMCA}
  Let $\eta \in(0,1]$. We assume that the assumptions of Theorem~\ref{thm_biais_variance_nested} hold and besides that
 \begin{equation*}
    \forall p \in \{2,\dots,P\},\ \E\left[ \frac{D_{2+\eta}(X)}{|E^p_X-M^{p-1}_X|^\eta} \phi^2(X) \right] <\infty,
 \end{equation*}
 where $D^p_{2+\eta}(X)=\E[|Y^p-\E[Y^p|X]|^{2+\eta}|X]$. Then, by taking when $\varepsilon\to 0$
  \begin{equation}\label{reglage_multi_Anti}
    L=\left\lceil  \frac{2}{1+\eta} \frac{|\log(\varepsilon)|}{\log(2)}\right\rceil,  J_0=2^{\left\lceil\frac{ 2|\log(\varepsilon)|}{\log(2)}  \right\rceil}=O(\varepsilon^{-2}) \text{ and } J_l=\lceil J_02^{-\left(1+\frac \eta 4\right)l} \rceil, l\in\{1,\dots L\},
  \end{equation}
we have $MSE(\widehat{I}^{MLMC}_A)=\E[(\widehat{I}^{MLMC}_A-I)^2]=O(\varepsilon^2)$ with a computational cost in $O(\varepsilon^{-2})$.
\end{theorem}
\begin{proof}
  We have $\text{bias}(\widehat{I}_A^{MLMC})=\text{bias}(\widehat{I}^{MLMC})=O(2^{-\frac{1+\eta}2 L})$. By   Proposition~\ref{prop_antithetic}, the variance of each level satisfies
  $$ \V\left(\left(\widehat{M}^P_{K_l}-\frac{\widehat{M}^P_{K_{l-1}}+\widehat{M}^{P,\prime}_{K_{l-1}}}2\right) \phi(X) \right)=O(K_l^{-\left(1+\frac \eta 2 \right)})=O(2^{-l\left(1+\frac \eta 2 \right)}),$$
  and the computational cost of $\left(\widehat{M}^P_{K_l}-\frac{\widehat{M}^P_{K_{l-1}}+\widehat{M}^{P,\prime}_{K_{l-1}}}2\right) \phi(X) $ is in $O(K_l)=O(2^l)$. We are thus in the framework of Theorem~1 of~\cite{Giles}, and we just check that the choice of parameters~\eqref{reglage_multi_Anti} gives the claim. By using the bias variance decomposition, we have
  $$ MSE(\widehat{I}^{MLMC}_A) \le C\left( 2^{-(1+\eta) L} + \frac {1}{J_0} + \sum_{l=1}^L \frac 1 {J_lK_l^{1+\frac \eta 2}} \right)\le C\left(\varepsilon^2 +\varepsilon^2 \sum_{l=0}^L 2^{-\frac \eta 4 l}\right). $$
  Since $\sum_{l=0}^L 2^{-\frac \eta 4 l}\le \sum_{l=0}^\infty 2^{-\frac \eta 4 l}=\frac 1 {1- 2^{-\frac \eta 4}}$, we indeed have $MSE(\widehat{I}^{MLMC}_A)=O(\varepsilon^2)$.
 Observing that for $\varepsilon \in \R_+^*$ small enough, we have $J_02^{-\left(1+\frac \eta 4\right)L}\ge 1$ and thus $J_l\le 2 J_0 \times 2^{-\left(1+\frac \eta 4\right)l}$ for $l\in \{0,\dots,L\}$, we can upper bound the computational cost as follows
  $$\sum_{l=0}^L J_lK_l \le 2 J_0 K_0  \sum_{l=0}^L 2^{-\frac \eta 4 l} \le \frac {4 K_0 \varepsilon^{-2}} {1- 2^{-\frac \eta 4}}. \qedhere$$
\end{proof}

\begin{remark}We can easily extend Theorem~\ref{thm_MLMCA} if we assume that the assumption of Theorem~\ref{thm_biais_variance_nested} is true for some $\eta_1\in (0,1]$ and that
    \begin{equation*}
    \forall p \in \{2,\dots,P\},\ \E\left[ \frac{D^p_{2+\eta_2}(X)}{|E^p_X-M^{p-1}_X|^{\eta_2}} \phi^2(X) \right]<\infty,
    \end{equation*}
    for some $\eta_2>0$. If we then take
    $$     L=\left\lceil  \frac{2}{1+\eta_1} \frac{|\log(\varepsilon)|}{\log(2)}\right\rceil,  J_0=2^{\left\lceil\frac{ 2|\log(\varepsilon)|}{\log(2)}  \right\rceil}=O(\varepsilon^{-2}) \text{ and } J_l=\lceil J_02^{-\left(1+\frac {\eta_2} 4\right)l} \rceil, l\in\{1,\dots L\},$$
    we get in the same way that $MSE(\widehat{I}^{MLMC}_A)=O(\varepsilon^2)$ with a computational cost in $O(\varepsilon^{-2})$. However, roughly speaking, the integrability assumption of Theorem~\ref{thm_biais_variance_nested} for the bias deals with the integrability of $\frac{1}{|E^p_X-M^{p-1}_X|^{\eta_1}}$ when $|E^p_X-M^{p-1}_X|$ is close to~$0$, similarly as the assumption for the variance estimate. Thus, it is rather natural to consider $\eta_1=\eta_2$, and  we state Theorem~\ref{thm_MLMCA} in this case for sake of simplicity.
\end{remark}

\begin{remark} In all the presentation of the MLMC method, we have considered $K_l=K_02^l$. We could more generally consider $K_l=K_0\nu^l$, with $\nu\ge 2$. In this case, the antithetic estimator~\eqref{Antithetic_MLMC_estim} would then be defined by
  $$\widehat{I}^{MLMC}_A=\frac{1}{J_0}\sum_{j=1}^{J_0}\widehat{M}^P_{0,j,K_0}\phi(X_{0,j})+\sum_{l=1}^L \frac{1}{J_l}\sum_{j=1}^{J_l} \left( \widehat{M}^P_{l,j,K_l} -\frac{1}{\nu} \sum_{\tilde{\nu}=1}^\nu\widehat{M}^{P,\tilde{\nu}}_{l,j,K_{l-1}}\right)\phi(X_{l,j}),$$
where  $\widehat{M}^{P,\tilde{\nu}}_{l,j,K_{l-1}}$ is the estimator of the maximum obtained with the samples $Y^p_{l,j,k}$ with $(\tilde{\nu}-1) K_{l-1}+1 \le k \le \tilde{\nu} K_{l-1}$. In the context of MLMC for Stochastic Differential Equations, Giles~\cite[Section 4.1]{Giles1} proposes a heuristic method to determine $\nu$ and uses $\nu=4$ in his experiments. Tuning the parameter~$\nu$ will not improve the asymptotic rate of convergence in $O(\varepsilon^{-2})$ given by Theorem~\ref{thm_MLMCA}, but may improve the convergence by a multiplicative factor. For simplicity and clarity, we have kept $\nu=2$ through all the paper and we leave the optimization in~$\nu$ for further research.   
\end{remark}

\subsection{Least Squares Monte Carlo techniques for Nested Expectations}

In this paragraph, we aim at presenting briefly the classical technique of regression in our context, i.e. for the calculation of~$I$. For simplicity, we only consider here regressors that are indicator functions. 

Let $N_r \in \N^*$ be the number of regressors. We consider $B_1,\dots,B_{N_r}\in \mathcal{G}$ disjoint measurable sets of the space where~$X$ takes values, and we define for $n\in \{1,\dots, N_r\}$ and $p \in \{1,\dots, P\},$
$$ \alpha^p_n = \E[Y^p|X\in B_n]=\frac{\E[Y^p \mathds{1}_{X \in B_n}]}{\Pb(X \in B_n)} \quad (\text{with the convention } 0/0=0).$$
Then we have
$$ \forall p \in \{1,\dots,P\}, \ \E\left[ \left(Y^p-\sum_{n=1}^{N_r} \alpha^p_n  \mathds{1}_{X \in B_n} \right)^2\right]= \min_{\alpha_1,\dots,\alpha_{N_r}\in \R} \E\left[ \left(Y^p-\sum_{n=1}^{N_r} \alpha_n  \mathds{1}_{X \in B_n} \right)^2\right], $$
i.e. $\sum_{n=1}^{N_r} \alpha^p_n  \mathds{1}_{X \in B_n}$ is the $L^2$ projection of~$Y^p$ on~$\{\sum_{n=1}^{N_r} \alpha_n  \mathds{1}_{X \in B_n} : \alpha_1,\dots,\alpha_{N_r}\in \R\}$. It is a natural proxy of~$E^p_X$, which is the $L^2$ projection on the larger space of $\sigma(X)$-measurable random variables. We then define $\gamma^P_n=\max_{p=1,\dots,P}\alpha^p_n$, so that ~$\sum_{n=1}^{N_r} \gamma^P_n  \mathds{1}_{X \in B_n}$ approximates $M^P_X$.

Let us consider $(X_j,Y_j)_{1\le j\le J}$ an i.i.d. sample following the distribution of $(X,Y)$. We define
$$ \widehat{\alpha}^p_{n,J} = \frac{\sum_{j=1}^J Y^p_j \mathds{1}_{X_j \in B_n}}{\sum_{j=1}^J \mathds{1}_{X_j \in B_n}}
\quad (\text{with the same convention } 0/0=0)$$
and have similarly
$$\forall p \in \{1,\dots,P\}, \ \frac 1 J \sum_{j=1}^J \left(Y^p_j-\sum_{n=1}^{N_r}\widehat{\alpha}^p_{n,J}  \mathds{1}_{X_j \in B_n} \right)^2= \min_{\alpha_1,\dots,\alpha_{N_r} \in \R}\frac 1 J \sum_{j=1}^J \left(Y^p_j-\sum_{n=1}^{N_r} \alpha_{n}  \mathds{1}_{X_j \in B_n} \right)^2.  $$
We define $ \widehat{\gamma}^P_{n,J}=\max_{p=1,\dots,P}\widehat{\alpha}^p_{n,J}$, so that $\sum_{n=1}^{N_r}\widehat{\gamma}^P_{n,J}  \mathds{1}_{X_j \in B_n}$ approximates $M^P_{X_j}$. Thus, we define the Least Squares Monte-Carlo estimator of~$I$ by
\begin{equation}\label{def_ILSMC}\widehat{I}^{LSMC}=\frac 1 J \sum_{j=1}^J  \phi(X_j) \sum_{n=1}^{N_r} \widehat{\gamma}^P_{n,J}  \mathds{1}_{X_j \in B_n}. 
\end{equation}
We are interested in estimating the MSE of this estimator. The next proposition gives a framework to analyse it, which is useful to determine asymptotically the number of regressors and the number of Monte-Carlo samples to reach a given precision~$\varepsilon>0$.

\begin{proposition}\label{prop_MSE_LSMC}
  For $p=1,\dots,P$, we set $\sigma_p(x)=\V(Y^p|X=x)$, and assume that there exists $\overline{\sigma},\overline{\phi}\in \R_+^*$ such that for all $x\in G$, $\sigma^2_p(x)\le \overline{\sigma}^2$ and $|\phi(x)|\le \overline{\phi}$. Then, we have
  $$ \E[(\widehat{I}^{LSMC}-I)^2] \le 2 \overline{\phi}^2 \left( \frac{\overline{\sigma}^2N_rP+ \E[(M^P_X)^2]}J + \sum_{p=1}^P \E\left[\left( E^p_X -\sum_{n=1}^{N_r}\alpha^p_{n}  \mathds{1}_{X \in B_n}  \right)^2\right]\right).$$
  We now suppose in addition that $X$ takes values in $G=[0,1]^d$ and that:
  \begin{enumerate}
  \item $N_r=n_r^d$ for some $n_r\in \N$, and for any $n\in \{1,\dots,N_r\}$,
    $$B_n=\left[\frac{i_1}{n_r}, \frac{i_1+1}{n_r}\right)\times\dots\times \left[\frac{i_d}{n_r}, \frac{i_d+1}{n_r}\right).$$
      where $i_1,\dots,i_d \in \{0,\dots,n_r-1\}$ are the unique integers determined by the decomposition in base~$n_r$ of $n-1$, i.e. $n-1=\sum_{k=1}^d i_k n_r^{k-1}$,
  \item for any $p\in\{1,\dots,P\}$, the function $G \ni x \mapsto E^p_x=\E[Y^p|X=x]$ is Lipschitz continuous with constant $L$ for the $\|\|_\infty$ norm on $\R^d$. 
  \end{enumerate}
  Then, we have
  $$  \E[(\widehat{I}^{LSMC}-I)^2] \le 2 \overline{\phi}^2 \left( \frac{\sigma^2 N_r P }J + \frac{PL^2}{N_r^{2/d}} \right). $$
With this upper bound, taking $J\sim c \varepsilon^{-d-2}$ and $N_r\sim c' \varepsilon^{-d}$ for some constants $c,c'>0$ is an asymptotic optimal choice to have  $\E[(\widehat{I}^{LSMC}-I)^2] =O(\varepsilon^2)$, with an overall computational cost in $O(\varepsilon^{-d-2})$.
\end{proposition}
\noindent In comparison with the MLMC estimator, it is worth to notice that $\widehat{I}^{LSMC}$ suffers from the curse of dimensionality. This is a well-known weakness of Least Squares Monte Carlo, see e.g.~\cite[Section 8.2]{Gobet}, which is not related to our particular problem. The larger is the dimension of~$G$ (the space where $X$ takes values), the more it requires computational effort. As we will see, for the problem of the calculation of the SCR for ALM management, this is particularly detrimental. The MLMC antithetic estimator~\eqref{Antithetic_MLMC_estim} has the clear advantage to converge with the same computational cost $O(\varepsilon^{-2})$ as an unbiased Monte-Carlo estimator, independently from the dimension of the problem. 

\begin{proof}
From the definition of $\widehat{I}^{LSMC}$~\eqref{def_ILSMC}, we get by using $(a+b)^2\le 2a^2 +2b^2$ and Jensen's inequality: 
\begin{align*}
  \E[(\widehat{I}^{LSMC}-I)^2]&=\E\left[\left(\widehat{I}^{LSMC}-\frac 1 J \sum_{j=1}^J  \phi(X_j) M^P_{X_j} +\frac 1 J \sum_{j=1}^J  \phi(X_j)M^P_{X_j} -I\right)^2\right] \\
  &\le \E\left[ \frac 2 J \sum_{j=1}^J  \phi^2(X_j) \left(\sum_{n=1}^{N_r} \widehat{\gamma}^P_{n,J}  \mathds{1}_{X_j \in B_n} -M^P_{X_j}\right)^2\right]+2\frac{\V(\phi(X)M^P_X)}{J} \\
  &\le 2 \overline{\phi}^2\left( \E \left[\frac 1 J \sum_{j=1}^J   \left(\sum_{n=1}^{N_r} \widehat{\gamma}^P_{n,J}  \mathds{1}_{X_j \in B_n} -M^P_{X_j}\right)^2\right]+\frac{\E[(M^P_X)^2]}{J}\right). 
\end{align*}
Now, Theorem~8.2.4~\cite{Gobet} gives 
  $$  \E\left[ \frac 1J \sum_{j=1}^J  \left( E^p_{X_j} -\sum_{n=1}^{N_r}\widehat{\alpha}^p_{n,J}  \mathds{1}_{X_j \in B_n}  \right)^2\right] \le \frac{\sigma^2 N_r}J +\E\left[\left( E^p_X -\sum_{n=1}^{N_r}\alpha^p_{n}  \mathds{1}_{X \in B_n}  \right)^2\right]. $$
We recall that $\widehat{\gamma}^P_{n,J}=\max_{p=1,\dots,P}\widehat{\alpha}^p_{n,J}$ and observe that for $X_j\in B_n$, we have
  $$(M^P_{X_j}-\max_{p=1,\dots,P}\widehat{\alpha}^p_{n,J})^2\le \max_{p=1,\dots,P}(E^p_{X_j}-\widehat{\alpha}^p_{n,J})^2\le \sum_{p=1}^P(E^p_{X_j}-\widehat{\alpha}^p_{n,J})^2$$
since $|\max_{p=1,\dots,P}a_p -\max_{p=1,\dots,P}b_p |\le \max_{p=1,\dots,P}|a_p-b_p|$ for any $a,b\in\R^P$. This gives the first upper bound.

We now consider the case $G=[0,1]^d$ with the related assumptions. Then, for $X\in B_n$, we have for any $p$
\begin{align*}
  |E^p_X-\alpha^p_n|&=\left|E^p_X-\int_{x\in B_n}E^p_x\Pb(X\in dx|X \in B_n)\right|\\
  &\le \int_{x\in B_n} |E^p_X-E^p_x|\Pb(X\in dx|X \in B_n)\le \frac{L}{n_r}=\frac{L}{N_r^{1/d}},
\end{align*}
since $\|X-x\|_\infty\le \frac 1 {n_r}$ for $X,x\in B_n$. This gives the second bound. To have this upper bound smaller than $C \varepsilon^2$ for some constant $C>0$, one must at least have $N_r\ge c_1\varepsilon^{-d}$ and $J\ge c_2 N_r\varepsilon^{-2}$ for some constants $c_1,c_2>0$, which leads to take $N_r\sim c' \varepsilon^{-d}$ and $J \sim \varepsilon^{-d-2}$.

Last, we observe that the computational cost to find $n$ such that $x\in B_n$ is constant since $i_k=\lfloor n_r x_k\rfloor$ and $n=1+\sum_{k=1}^di_kn_r^{k-1}$. Therefore, computing all the $2N_r$ sums $\sum_{j=1}^J Y^p_j \mathds{1}_{X_j \in B_n}$ and  $\sum_{j=1}^J \mathds{1}_{X_j \in B_n}$ that define $ $ can be achieved with a computational cost of $O(J)$, and the calculation of~\eqref{def_ILSMC} costs similarly $O(J)$. Since $J\sim c \varepsilon^{-d-2}$, we get the claim.
\end{proof}

\subsection{Numerical results on a toy example: the Butterfly Call Option with the Black-Scholes model}



The goal of this section is to illustrate the theoretical results on a simple case where the conditional expectations are known explicitly. Thus, we consider an asset following the Black-Scholes model:
$$S_t=S_0\exp\left(\sigma W_t-\frac{\sigma^2}2t\right), \ t\ge 0,$$
where $W$ is a standard Brownian motion and $\sigma>0$ is the volatility. 
We consider a butterfly option with payoff at time $T>0$:
\begin{align*}
\psi(S_T)=(S_T-K_1)^++(S_T-K_2)^+-2\left(S_T-\frac{K_1+K_2}{2}\right)^+,
\end{align*}
where $0<K_1<K_2$. The price of this butterfly option at time $t\in[0,T]$ is given by
$$\E[\psi(S_T)|S_t]=\mathrm{Call}^{\mathrm{BS}}(T-t,S_t,K_1)+\mathrm{Call}^{\mathrm{BS}}(T-t,S_t,K_2)-2\mathrm{Call}^{\mathrm{BS}}\left(T-t,S_t,\frac{K_1+K_2}2\right)=:\mathrm{Butterfly}(T-t,S_t),$$
with $\mathrm{Call}^{\mathrm{BS}}(t,s,K)=s \mathcal{N}(\frac 1 {\sigma \sqrt{t}}\ln(s/K) + \frac \sigma 2 \sqrt{t} )- K\mathcal{N}(\frac 1 {\sigma \sqrt{t}}\ln(s/K) - \frac \sigma 2 \sqrt{t} )$, where $\mathcal{N}$ is the cumulative distribution function of the standard normal distribution.

Now, we consider multiplicative upward and downward shocks $s^{up/down}$ on the asset value that occur instantaneously at time~$t$. We want to compute the worst loss between these shocks when it is positive. Since the Black-Scholes model is multiplicative with respect to the spot value, these shocks amount to multiply the asset by $1+  s^{up}$ and $1+s^{down}$, with $-1<s^{down}<0<s^{up}$.  Hence, setting $X=S_t$, $Y^1=\left(\psi(S_T)- \psi((1+s^{up})S_T)\right)$ and $Y^2=\left(\psi(S_T)- \psi((1+s^{down})S_T) \right)$  we want to compute the following quantity :
\begin{equation*}
I=\E\left[\max\left\{\E[Y^1|X],\E[Y^2|X],0 \right\}\right].
\end{equation*}
We are thus indeed in our general framework with $P=3$ and $Y^3=0$ and $\phi(x)=1$, and we have
$$M^3_X=\max\left\{\mathrm{Butterfly}(T-t,(1+s^{up})X), \mathrm{Butterfly}(T-t,(1+s^{down})X) ,0\right\}.$$
Since $X$ follows a log-normal distribution, the exact value of $I$ can be thus obtained by numerical integration.

{\bf Numerical values.} In all our numerical experiments, we consider the initial price $S_0=100$, the volatility $\sigma=0.3$, the strikes $K_1=S_0+a$ and $K_2=S_0-a$ with $a=50$, the option maturity $T=2$ years and perform the shocks at $t=1$ year. In our tests, we take $s^{up}=0.2$ and $s^{down}=-0.2$.

Figure~\ref{fig:biais} illustrates the bias $\E[\widehat{M}^3_K-M^3_X]$ in function of $K$ with a log-scale. The expectation is approximated by the Nested Monte-Carlo with $J=10^4$ to get a negligible statistical error. As a comparison, the function $K\mapsto 1/K$ is drawn, and we observe that two curves are quite parallel, which indicates that the bias behaves asymptotically like $c/K$. Also, we have drawn in Figure~\ref{fig:variance} the variance of $\widehat{M}^3_{K_l}-\frac{\widehat{M}^3_{K_{l-1}}+\widehat{M}^{3,\prime}_{K_{l-1}}}2$ in function of~$K_l$, and we observe a behaviour in $K_l^{-3/2}$. Thus, it is reasonable to apply then the Multilevel method with $\eta=1$ to determine the parameters in Equation~\eqref{reglage_multi_Anti}. We have drawn in Figure~\ref{fig:cv_MLMC_Anti} the RMSE in function of the computational cost (defined by $\sum_{l=0}^LJ_lK_l$) for different values of $\eta\le 1$. We observe a behaviour in $\varepsilon^{-2}$, which is in line with Theorem~\ref{thm_MLMCA}. The RMSE is calculated empirically, and we have runned many times the MLMC estimator to do so. 
\begin{figure}[h!]
 \begin{minipage}[t]{0.48\textwidth}
   \includegraphics[width=\textwidth]{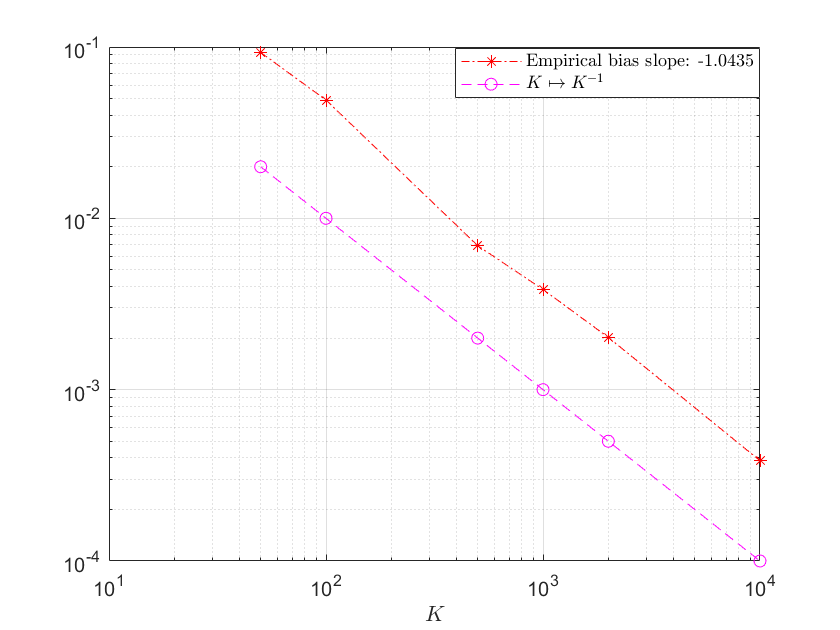}
\caption{Bias behaviour of the nested estimator\label{fig:biais}} 
 \end{minipage}
 \
 \begin{minipage}[t]{0.48\textwidth}
    \includegraphics[width=\textwidth]{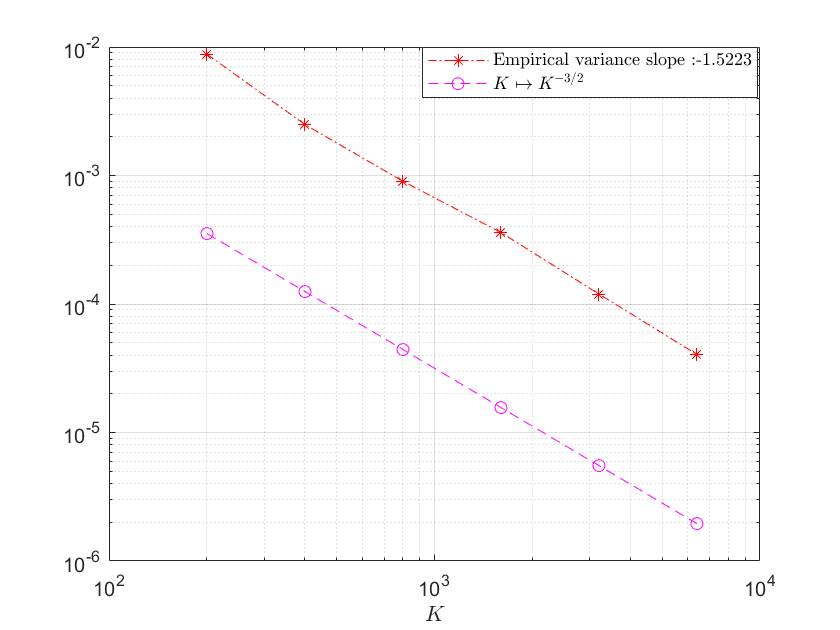}
    \caption{ $\V \left(\widehat{M}^3_{K_l}-\frac{\widehat{M}^3_{K_{l-1}}+\widehat{M}^{3,\prime}_{K_{l-1}}}2\right)$ in function of~$K_l$\label{fig:variance}}
  \end{minipage}
\end{figure}

\begin{figure}[h!]
 \begin{minipage}[b]{0.48\textwidth}
    \includegraphics[width=\textwidth]{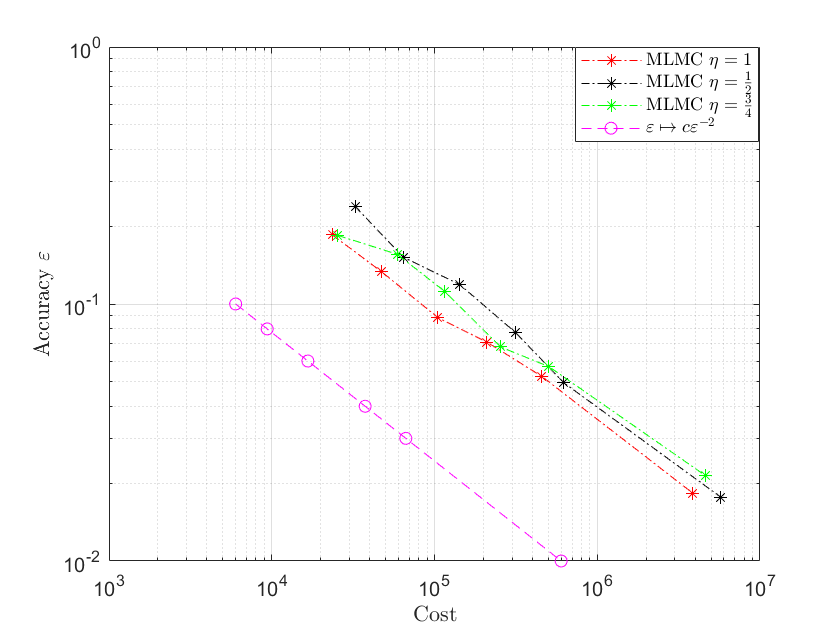}
    \caption{Empirical RMSE of the MLMC antithetic estimator as a function of the cost (log-scale) for different values of~$\eta$. \label{fig:cv_MLMC_Anti}}
 \end{minipage}
 \ 
  \begin{minipage}[b]{0.48\textwidth}
    \includegraphics[width=\textwidth]{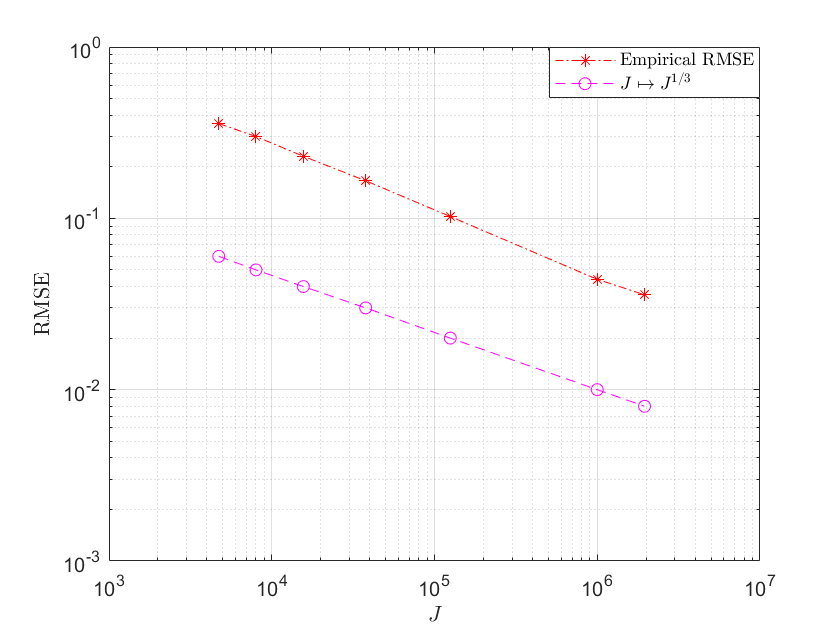}
    \caption{Empirical RMSE of the LSMC  estimator  as a function of the number of samples $J$ (log-scale) \label{fig:cv_LSMC}}
    \end{minipage}
\end{figure}
We now present the implementation of LSMC estimator. We note that $X'=\frac 1 \sigma \log(X)+\sigma/2$ is a standard normal distribution. It therefore takes with a probability greater than 99\% its values in $[-3,3]$. We notice that $\E[Y^p|X]=\E[Y^p|X']$ and take the following regressors
\begin{equation*}
\mathds{1}_{X'\in[-3+6j/N_r,-3+6(j+1)/N_r]},\quad j=0,\ldots N_r-1.
\end{equation*}
Up to a translation, we are thus in the framework of Proposition~\ref{prop_MSE_LSMC}. In Figure~\ref{fig:cv_LSMC}, we have plotted the RMSE as a function of the number of samples~$J$, which is also the computational cost of the method. The behaviour is in line with the theoretical result given by Proposition~\ref{prop_MSE_LSMC} with $d=1$.

We already see on this one-dimensional example that the MLMC estimator has some benefit in terms of convergence with respect to the LSMC estimator. As we will see in the next section  for the $SCR$ estimation, this benefit is much more  important when $X$ takes values in a high-dimensional space.

\section{Calculation of the SCR with the Standard Formula in an ALM model}\label{Sec_ALM}

In this section, we want to illustrate and compare the MLMC and LSMC methods on a more realistic example for the application to insurance. Namely, we consider the case of Asset Liability Management (ALM) for life insurance contracts. We are interested in computing the SCR with the standard formula after $t$ years. This example is of practical interest and the conditional expectations that are at stake are typically high-dimensional. In fact, the process that determines the ALM strategy is really path-dependent and involves book values, market values, crediting rates, etc. To illustrate this, we will use the recent ALM model that we have developed in~\cite{AlChIn}. In this section, we mainly focus on the interest rate module of the SCR with the standard formula, but discuss also the calculation of the SCR on the market risk in the last subsection. 

\subsection{The ALM model in a nutshell}\label{subsec_ALM}

In this section, we briefly present the ALM model developed in~\cite{AlChIn} and refer to this paper for the full details. We consider an insurance company that handles a life insurance, namely a General Account guaranteed with participation contracts. We consider a runoff portfolio with an initial Mathematical Reserve $MR_0$ corresponding to policyholders' deposit. The Capitalization Reserve (a buffer for capital gains on bonds imposed by the French legislation) and the Profit Sharing Reserve (a buffer for capital gains on stocks to smooth the crediting rate) are empty at time~$0$, i.e. $CR_0=0$ and $PSR_0=0$. At time~$0$, the insurance company invests $MR_0$ in two asset classes, stocks and riskless bonds, with respective weights $w^S_0\in[0,1]$ and $w^b_0=1-w^S_0$.  Thus, the initial Market Value and Book Value in stock (resp. in bonds) is given by
$$MV^S_0=BV^S_0=w^S_0MR_0 \text{ (resp. } MV^b_0=BV^b_0=w^b_0MR_0 ). $$
During all the ALM strategy, the insurance company invests in an equity asset  $(S_t)_{t\ge 0}$ that may be a stock index or more generally an average of stocks with weights corresponding to the investment of the insurance company. It therefore has $\phi^S_0=MV^S_0/\phi^S_0$ equity assets at time~$0$. The insurance company also invests in bonds, and we assume that this investment is made with an equally weighted portfolio of bonds with maturities $1,\dots,n$. We introduce some notation to precise this: we denote $P(t,t+i)$ the price of a Zero-Coupon bond at time~$t$ with maturity $t+i$, $B(t,n,c)=\sum_{i=1}^n c P(t,t+i)+P(t,t+n)$ the price at time~$t$ of a bond with constant coupon $c\in \R$, unit nominal value and maturity~$n$ and for ${\bf c}=(c^i)_{i\in \{1,\dots,n\}} \in \R^n$ we denote by
$$ \bar{B}(t,n,{\bf c})=\frac 1 n \sum_{i=1}^n B(t,i,c^i)$$
the value of an equally weighted portfolio of bonds with coupons~${\bf c}$. During all the ALM strategy, we assume that bonds are bought at par with the swap rate $c_{swap}(t,n)=\frac{1-P(t,t+n)}{\sum_{i=1}^n P(t,t+i)}$. We set ${\bf c}_0=(c_{swap}(0,i))_{i\in\{1,\dots,n\}}$ and have $\bar{B}(0,n,{\bf c}_0)=1$. At time~$0$, the insurance company has then $\phi^b_0=MV^b_0=MV^b_0/\bar{B}(0,n,{\bf c}_0)$ assets  $\bar{B}(0,n,{\bf c}_0)$.

We assume that the portfolio is handled up to time $T\in \N^*$ and that it is static on each period $(t-1,t)$, $t\in \{1,\dots,T\}$. At each time~$t$, it is reallocated in such a way to have at the end of the reallocation
$$\phi^S_t=\frac{w^S_t MV_t}{S_t}, \ \phi^b_t=\frac{w^b_t MV_t}{\bar{B}(t,n, {\bf c }_t)} $$
quantities of equity assets and bonds, where $MV_t$ denotes the market value of the portfolio at time~$t$ and $w^S_t=1-w^b_t\in[0,1]$ is the target weight decided for ALM strategy. The coupons ${\bf c}_t \in \R^n$ are determined by the reallocation procedure that we describe now and takes into account the specificities of life insurance contracts. We decompose this reallocation in five steps:
\begin{enumerate}
\item {\em Calculation of the cash inflows and book value movements related to the bonds.} Since the portfolio composition is unchanged on $(t-1,t)$, the insurer receives $\frac{\phi^b_{t-1}}n \left(1 + \sum_{i=1}^nc^i_{t-1} \right)$ corresponding to the nominal value of the expiring bonds and the coupons. The value of the matured bonds~$\frac{\phi^b_{t-1}}n$ is removed from the book value of bonds $BV^b_t$.
\item {\em Payment of the policyholders that exit their contract.} The proportion of policyholders that exit on $[t-1,t]$ is denoted by~$p^e_{t-1}$. It is modelled as the sum of a deterministic part related to the relevant life table and of a dynamic part modelling surrenders $DSR(\Delta_{t-1})=DSR_{max}\mathds{1}_{\Delta_t\le \alpha}+DSR_{max}\frac{\beta-\Delta_t}{\beta-\alpha} \mathds{1}_{\alpha< \Delta_t<\beta}$, where $\Delta_{t-1}$ is the difference between the crediting rate to policyholders $r_{ph}(t-1)$ and a competitor rate $r^{comp}_{t-1}$. We assume that policyholders exit uniformly on  $[t-1,t]$, and the amount to pay is thus $p^e_{t-1}MR_{t-1}(1+r_G/2)$, where $r_G$ is the minimum guaranteed rate. This means that they are remunerated with this rate on the last period.
\item {\em Reallocation step.} At this step, the market value of the portfolio is given by
  $$MV_t=G_t+\phi^S_{t-1}S_t + \frac{\phi^b_{t-1}}{n}\sum_{i=1}^{n-1}B(t,i,c^{i+1}_{t-1}),  $$
  where $G_t$ is the liquidity gap that corresponds to the difference between the cash inflows and outflows of the two first steps. The second term $\phi^S_{t-1}S_t$ represents the market value of equity assets, and the last term the market value of bond assets. Note that a bond at time~$t-1$ with maturity~$i+1$ and coupon $c^{i+1}_{t-1}$ becomes at time~$t$ a bond with maturity~$i$ with the same coupon.

  The portfolio is reallocated with the prescribed weights~$w^S_t\in [0,1]$ and $w^b_t=1-w^S_t$ given by the ALM strategy. The amount of equity assets to hold is thus given by $\phi^S_t=w^S_tMV_t/S_t$. If this quantity is greater than $\phi^S_{t-1}$, there is a purchase of $\phi^S_t-\phi^S_{t-1}$ equity assets which increases the book value $BV^S_t$ by $(\phi^S_t-\phi^S_{t-1})S_t$. If this quantity is lower than $\phi^S_{t-1}$, there is a sell of $\phi^S_{t-1}-\phi^S_{t}$ equity assets which decreases the book value $BV^S_t$ by the factor $\phi^S_t/\phi^S_{t-1}$. This generates capital gain or loss on stocks that is registered, since capital gain has to be redistributed to policyholders with a participation rate $\pi_{pr}\in [0,1]$.

  The reallocation in bonds follows the same principles but is more involved. At the end of this step, the portfolio in bonds is made with $\phi^b_t$ combinations of bonds~$\bar{B}(t,n,{\bf c}_t)$. Since the bonds are bought at par, there is a precise relation between ${\bf c}_t$, ${\bf c}_{t-1}$ and the swap rates at time~$t$. According to the French legislation rules, the capital gain or loss on bonds is stored in the Capitalization Reserve and is separated from the ALM portfolio. Details can be found in~\cite{AlChIn}.

\item {\em Determination of the crediting rate.} This step determines the policyholders' earning rate $r_{ph}(t)$ on the period $(t-1,t)$. Due to regulatory constraint, it has to be greater than the minimum guaranteed rate $r_G$ and the amount distributed to policyholders has to be greater than the proportion $\pi_{pr}$ of the gains (participation rate). Besides the insurance company compares $r_{ph}(t)$ with a competitor rate $r^{comp}_t$ (typically the market short rate) and tries at best to have $r_{ph}(t)\ge r^{comp}_t$ to avoid dynamic surrenders. We call ``target rate'', the maximum rate given by these three constraints. 

  The amount to distribute is typically made with the coupons, the capital gain or loss on stocks and possibly dividends. To smooth these gains along the years, the insurance company uses a Profit Sharing Reserve. In addition, the insurance company may also want to realize a part of latent gain or loss on stocks. In the model developed in~\cite{AlChIn}, the amount to distribute depends on all these quantities. We distinguish four cases (from the best to the worst) to determine $r_{ph}(t)$.
  \begin{enumerate}[label=(\Alph*)]
  \item The target rate can be distributed without using latent gain or by realizing all the latent loss on stocks.
  \item The target rate can be distributed by using latent gain or without realizing all the latent loss on stocks. The proportion of gain or loss is determined accordingly.
  \item The target rate cannot be reached with the available amount, but the minimum guaranteed rate can be distributed. The insurance company then uses all the latent gain in order to serve the best possible rate.
  \item The minimum guaranteed rate cannot be reached with the available amount. Then, the insurance company clears out the Profit Sharing Reserve and credit the policyholders with the lowest rate above $r_G$ that also satisfies participation rate constraints. 
  \end{enumerate}
  Once $r_{ph}(t)$ is determined, the Mathematical Reserve of the remaining policyholders is updated accordingly: $MR_t=MR_{t-1}(1-p^e_{t-1})(1+r_{ph}(t))$. The Profit Sharing Reserve and the book value of stocks are also modified according to the case. The shareholder's margin  can be calculated as well as the profit and loss $P\&L_t$ generated on the period $(t-1,t)$, which is defined as the sum of the shareholder's margin and the interest generated by the Capitalization Reserve.  Again, all the details can be found in~\cite{AlChIn}.

\item {\em Externalization of the Capitalization Reserve and of Shareholders' margin from the accounting.} This last step is a technical  accounting operation that slightly change the quantities of assets and the book values, while keeping unchanged the target weights $w^S_t$ and $w^b_t$.
\end{enumerate}
\noindent The last step at the final time~$T$ follows the same lines: instead of being reallocated, the portfolio is cleared and policyholders get back the remaining Mathematical Reserve $MR_T$.

\subsection{The Solvency Capital Requirement with the standard formula}\label{subsec_STD}

We now present the main lines of the SCR calculation with the standard formula as indicated by the EIOPA~\cite{EIOPA,EIOPA2}. Let us denote by $(\F_t,t\ge 0)$ the filtration representing the market information at time~$t\ge 0$ and $\Q$ the pricing measure. We consider a short-rate model $(r_t,t\ge 0)$ for interest rates and define at time~$t\in \{0,\dots, T-1\}$ the Basic Own Funds by
$$BOF_t=\E^\Q\left[\sum_{u=t+1}^T e^{-\int_t^u r_s ds}P\&L_u\bigg|\F_t\right],$$
i.e. the expected value of the discounted future profits and losses.
The principle of the standard formula is to apply shocks on each asset class (equity, interest rate, etc.) and evaluate the variation of Basic Own Funds. Then, the SCR on market risk is obtained by using a given formula that aggregates all risk modules. In this paper, we focus on the interest rate module, where upward and downward shocks are prescribed by the regulator. The methodology to apply these shocks is described in Section 2.5 of~\cite{AlChIn}. We have used in our simulations the shocks specified in~\cite{EIOPA}. At time~$t$, the SCR value of the interest module is then defined by
$$SCR^{int}_t=\max\{BOF_t-BOF^{\textup{upward shock}}_t,BOF_t-BOF^{\textup{downward shock}}_t ,0\},$$
where shocks are applied at time~$t$ on the interest-rate curve. We also set
$$SCR^{up}_t=\max\{BOF_t-BOF^{\textup{upward shock}}_t,0\} \text{ and } SCR^{down}_t=\max\{BOF_t-BOF^{\textup{downward shock}}_t,0\}, $$
so that $SCR^{int}_t=\max(SCR^{down}_t,SCR^{up}_t)$.
At time $t=0$, $SCR^{int}_0$ is a number that can be calculated by Monte-Carlo. This has been investigated in~\cite{AlChIn}. However, for the ALM strategy, it may be useful to have quantitative insights on the evolution of the SCR along the time to asses the cost of capital. Thus, in this paper we are interested with the valuation of
\begin{equation}\label{I_SCR}
  I=\E^\Pb[\max\{BOF_t-BOF^{\textup{upward shock}}_t,BOF_t-BOF^{\textup{downward shock}}_t ,0\}],
\end{equation}
the average value under the historical (or real) probability of the $SCR$ at time~$t$. If we denote by $\frac{d\Pb}{d\Q}\bigg|_{\F_t}$ the change of probability, we have
\begin{equation}\label{SCR_int_histo}
  I=\E^\Q\left[\frac{d\Pb}{d\Q}\bigg|_{\F_t} \max\{BOF_t-BOF^{\textup{upward shock}}_t,BOF_t-BOF^{\textup{downward shock}}_t ,0\}\right], 
\end{equation}
and we are precisely in the framework of Section~\ref{Sec_Math} if $X$ denotes a random variable that represents all the market information up to time~$t$ (i.e. $\sigma(X)=\F_t$).
The equity module of the SCR is similarly defined by
$$SCR^{eq}_t=\max\{BOF_t-BOF^{\textup{equity shock}}_t,0\},$$
where the equity shock amounts to a strong decrease of $S$ immediately after~$t$. Usually, the maximum with zero is useless since the shock is always negative. 
Last the standard formula that defines the SCR on market risk as follows (see Articles 164 and 165 of~\cite{DelReg}):
\begin{align}\label{def_SCR_mkt}
SCR^{mkt}_t=\sqrt{(SCR^{eq}_t)^2_t+(SCR^{int})^2_t+2\varepsilon SCR^{eq}_tSCR^{int}_t}
\end{align}
where $\varepsilon=0$ if the interest-rate exposure is due to the upward-shock on interest rates and $\varepsilon=\frac{1}{2}$ if it is due to the downward shock of the interest rate module. Thus, the expected value of the SCR is given by:
$$\E^{\Pb}[SCR^{mkt}_t]=\E^\Q\left[\frac{d\Pb}{d\Q}\bigg|_{\F_t}  \sqrt{(SCR^{eq}_t)^2_t+(SCR^{int})^2_t+2\varepsilon SCR^{eq}_tSCR^{int}_t}\right] .$$

\subsection{The stock and short-rate models}\label{subsec_MODEL}

We consider $(W_t,Z_t)_{t\ge 0}$ a standard two-dimensional Brownian motion under~$\Q$. Following~\cite{AlChIn}, we assume that the equity assets follows a Black-Scholes model and that the short interest rate follows a Vasicek++ (or Hull and White) model:
\begin{align}\label{model}
  \frac{dS_t}{S_t}&=r_tdt +\sigma_S dW_t\\
  r_t&=x_t+\varphi(t),\text{ with } dx_t=k(\theta-x_t)dt+\sigma_r (\gamma dW_t+\sqrt{1-\gamma^2}dZ_t),
\end{align}
where $\gamma \in [-1,1]$ tunes the dependence between equity and interest rates. We assume $k,\theta,\sigma_S,\sigma_r>0$. As explained in~\cite{AlChIn}, the shift function~$\varphi:\R_+\to \R$ is particularly convenient to implement the shocks prescribed by the EIOPA. Mainly, shocks amounts to modify the shift, leaving the dynamics of~$x$ unchanged, which makes easy to calculate the ALM strategies in both normal and shocked cases on each sample.

A new feature with respect to~\cite{AlChIn} is that we now also consider the dynamics under the real-world probability~$\Pb$. We assume here for simplicity the following basic change of probability
\begin{equation}\label{Chg_prob} \frac{d\Pb}{d\Q}\bigg|_{\F_t} =\exp \left( \lambda^W W_t+ \lambda^Z  Z_t -\frac12( (\lambda^W)^2+(\lambda^Z)^2 +2\gamma\lambda^W\lambda^Z
  )t \right)=:L_t,
\end{equation}
with $\lambda^W,\lambda^Z \in \R$. By the Cameron-Martin theorem, $dW^{\Pb}_t=dW_t-\lambda^W_tdt$ and $dZ^{\Pb}_t=dZ_t-\lambda^Zdt$ are independent Brownian motions under~$\Pb$. We then have the following dynamics under~$\Pb$:
\begin{align*}
  \frac{dS_t}{S_t}&=(r_t+\lambda^W \sigma_S)dt +\sigma_S dW^\Pb_t\\
  r_t&=x_t+\varphi(t),\text{ with } dx_t=k\left(\theta+\sigma_r \frac{\gamma \lambda^W+\sqrt{1-\gamma^2}\lambda^Z }{k}-x_t\right)dt+\sigma_r (\gamma dW^\Pb_t+\sqrt{1-\gamma^2}dZ^\Pb_t),
\end{align*}
To run this asset model with the ALM model described in Subsection~\ref{subsec_ALM}, we have to be able to sample $S_t$, $r_t$ and the change of probability $L_t$ at each time $t\in \N$. It is possible to do it exactly by using the following recurrence formula
\begin{align*}
  S_t&=S_{t-1}\exp \left(\int_{t-1}^tx_udu + \int_{t-1}^t\varphi(u)du  +\sigma_S(W_t-W_{t-1})-\frac{\sigma_S^2}2 \right),\\
  x_t&=x_{t-1}e^{-k}+\theta(1-e^{-k})+\sigma_r \int_{t-1}^te^{-k(t-u)}(\gamma dW_u+\sqrt{1-\gamma^2}dZ_u),\\
  L_t&=L_{t-1}\exp \left( \lambda^W (W_t-W_{t-1})+ \lambda^Z  (Z_t-Z_{t-1}) -\frac12( (\lambda^W)^2+(\lambda^Z)^2 +2\gamma\lambda^W\lambda^Z) \right),
\end{align*}
and 
$$\int_{t-1}^t x_udu =\frac1k (x_{t-1}-x_t)+ \theta +\frac{\sigma_r}k[\gamma (W_t-W_{t-1})+\sqrt{1-\gamma^2}(Z_t-Z_{t-1}) ].$$
The law of $(W_t-W_{t-1},Z_t-Z_{t-1}, \int_{t-1}^te^{-k(t-u)}(\gamma dW_u+\sqrt{1-\gamma^2}dZ_u))$ is a centered Normal distribution with covariance
$$\begin{bmatrix}1 & 0  & \gamma \frac{1-e^{-k}}{k } \\
0 & 1  &  \sqrt{1-\gamma^2} \frac{1-e^{-k}}{k }  \\
\gamma \frac{1-e^{-k}}{k } & \sqrt{1-\gamma^2} \frac{1-e^{-k}}{k }  & \frac{1-e^{-2k}}{2k }\\
\end{bmatrix}.$$
This is the same law as $\left(G_1,G_2,\gamma \frac{1-e^{-k}}{k } G_1+\sqrt{1-\gamma^2} \frac{1-e^{-k}}{k }G_2+\sqrt{\frac{1-e^{-2k}}{2k }-\left(\frac{1-e^{-k}}{k }\right)^2} G_3\right)$, where $G_1,G_2,G_3$ are independent standard Normal variables.
Once this triplet is sampled exactly, we can calculate easily~$(S_t,x_t,L_t)$ using the formulas above.

\subsection{Numerical experiments I: comparison between methods to calculate~$\E[SCR^{int}_t]$
}\label{subsec_NUMI}

We now present our numerical results for the computation of~$I$ defined by~\eqref{I_SCR}. We use the following parameters for the ALM model and for the asset model. They are summarized in Tables~\ref{Model_param} and~\ref{LM_param}. Unless specified, we also consider $\Pb=\Q$, i.e. that the real and risk-neutral probability are the same. We will discuss however later on the impact of this change of probability for the SCR. 
\begin{table}[H]
\centering  
\begin{tabular}{|l|r|}
  \hline
  Stock model & Short-rate model  \\
  \hline
  $S_0=1$ & $r_0=\theta=0.02$  \\
  $\sigma_S=0.1$ & $\sigma_r=0.01$  \\
   $\gamma=0$ & $k=0.2$  \\
  \hline
\end{tabular} 
\caption{Market-model parameters}\label{Model_param}
\end{table}

\begin{table}[H]
\centering  
\begin{tabular}{|l|l|}
  \hline
  Management Parameters & Liability Parameters \\
  \hline
  Target allocation in stock $w^s_t=0.05$ & Dynamic surrenders triggering thresholds $\beta=-0.01$ and $\alpha=-0.05$\\
  Target allocation in bond $w^b_t=0.95$ & Maximum lapse dynamic surrender rate $DSR_{max}=0.3$   \\
  Participation rate $\pi_{pr}=0.9$ &  Deterministic constant exit rate $\underline{p}=0.05$ \\
  Minimum guaranteed rate $r^G=0.015$ & Time horizon: $T=30$ years\\
  Competitor rate $r^{comp}_t=r_t$& \\
  Smoothing coefficient of the PSR: $\bar{\rho}=0.5$ &  \\
  Bond portfolio maximal maturity $n=20$ &  \\ 
  Projection Horizon $T=30$ &  \\ 
  \hline
\end{tabular} 
\caption{Liability and management parameters}\label{LM_param}
\end{table}
\begin{table}[H]
\centering  
\end{table}

The implementation of the MLMC antithetic estimator is easily made by using~\eqref{reglage_multi_Anti}. Instead, the implementation of the LSMC raises some issues. The main one is how to choose the regressors. In fact, the ALM model presented in Subsection~\ref{subsec_ALM} is truly  path-dependent, and one needs to know $(r_{t'},S_{t'})_{t'\in \{1,\dots,t\}}$ to determine the book values, the different reserves and the Bond portfolio at time~$t$.  Thus, $SCR^{int}_t$ depends on all the past before~$t$. For $t=1$ the dimension of the regression space is equal to~$2$ and the choice of the regressors $r_1$ and $S_1$ is obvious. When $t$ gets larger, this is no longer the case and in view of the theoretical complexity result of Proposition~\ref{prop_MSE_LSMC} one cannot afford to use all the $2t$ regressors. It is then important to select few regressors. We explain now the procedure that we have used.

\subsubsection{Selection of the regressors for the LSMC estimator}\label{parag_selec}
  In Table~\ref{risk_factor_description} we have listed $12$~relevant risk-factors for the insurance company. We will select the most relevant ones for the SCR interest rate module by using a forward selection procedure. 
 \begin{table}[h!]
 \centering
\begin{tabular}{ll}
\hline
\multicolumn{1}{|l|}{Attribute}  & \multicolumn{1}{l|}{Risk-factor description} \\ \hline
\multicolumn{1}{|l|}{$X_t^1=S_t$} & \multicolumn{1}{l|}{Equity asset value}                   \\ \hline
\multicolumn{1}{|l|}{$X_t^2=r_t$} & \multicolumn{1}{l|}{Short rate}                    \\ \hline
\multicolumn{1}{|l|}{$X_t^3=\phi^S_t$} & \multicolumn{1}{l|}{Position in Stock}       \\ \hline
\multicolumn{1}{|l|}{$X_t^4=\phi^b_t$} & \multicolumn{1}{l|}{Position in bonds}       \\ \hline
\multicolumn{1}{|l|}{$X_t^5=BV^b_t$} & \multicolumn{1}{l|}{Book value of bonds}       \\ \hline
\multicolumn{1}{|l|}{$X_t^6=BV^S_t$} & \multicolumn{1}{l|}{Book value of equity assets}       \\ \hline
\multicolumn{1}{|l|}{$X_t^7=MR_t$} & \multicolumn{1}{l|}{Mathematical Reserve}       \\ \hline
\multicolumn{1}{|l|}{$X_t^8=PSR_t$} & \multicolumn{1}{l|}{Profit sharing reserve}       \\ \hline
\multicolumn{1}{|l|}{$X_t^9=CR_t$} & \multicolumn{1}{l|}{Capitalization Reserve}       \\ \hline
\multicolumn{1}{|l|}{$X_t^{10}=MV_t$} & \multicolumn{1}{l|}{Portfolio market value }       \\ \hline
\multicolumn{1}{|l|}{$X_t^{11}=\phi^b_t\bar{B}(t,n,{\bf c_t})$} & \multicolumn{1}{l|}{Market value of bonds }       \\ \hline
\multicolumn{1}{|l|}{$X_t^{12}=\phi^S_tS_t$} & \multicolumn{1}{l|}{Market value of equity assets}       \\ \hline
\end{tabular}
\captionof{table}{Non exhaustive list of risk factors}\label{risk_factor_description}
 \end{table}
 
 To do so, we sample $J_v$ scenarios up to time~$t$ of the ALM model. This produces in particular $J_v$ samples of $(X_t^{1,j},\ldots,X_t^{12,j})_{j=1,\ldots,J_v}$. Then, we approximate for each scenario the value of the interest rate module of the SCR, $SCR_t^{int,j}$, by using a Nested Monte-Carlo with~$K$ of secondary scenarios. We note $\widehat{SCR}^{Nested,j}_t$ these approximations (we drop for readability the superscript ``$int$'' in Paragraph~\ref{parag_selec} since we only consider the interest rate module).  In our numerical application, we have taken $J_v=2000$ validation scenarios and $K=10^4$ inner scenarios. Let $\widehat{SCR}:\R^{12}\to \R$ be a function approximating the SCR from the values of~$X$.  We now consider the empirical RMSE, i.e.
\begin{align*}
\sqrt{\frac{1}{J_v}\sum_{i=1}^{J_v}(\widehat{SCR}^{Nested,j}_t -\widehat{SCR}(X^j_t))^2}
\end{align*} 
as a criterion to assess the accuracy of the regression function  $\widehat{SCR}$.

We start by selecting the first variable. Up to a linear rescaling of the sample we may assume without loss of generality that all the variables are in $[0,1]$. We consider the 12~possible regressor functions for $l\in \{1,\dots,12\}$,
$$\sum_{i=0}^{n_r-1} \widehat{\alpha}^l_i \mathds{1}_{X^l_t\in [i/n_r,(i+1)/n_r)}, \text{ with } \widehat{\alpha}^l_i=\frac{\sum_{j=1}^{J_v}\widehat{SCR}^{Nested,j}_t\mathds{1}_{X^{l,j}_t\in [i/n_r,(i+1)/n_r)}   }{\sum_{j=1}^{J_v} \mathds{1}_{X^{l,j}_t\in [i/n_r,(i+1)/n_r)}},$$
      and select $l_1^*\in \{1,\dots,12\}$ that achieves the lowest RMSE. Once $l_1$ is selected, we consider the following 11~regressor functions for $l\in \{1,\dots,12\} \setminus \{l_1 \}$:     
\begin{align*}
 & \sum_{i_1,i_2=0}^{n_r-1} \widehat{\alpha}^{l_1,l}_{i_1,i_2} \mathds{1}_{X^{l_1}_t\in [i_1/n_r,(i_1+1)/n_r)}\mathds{1}_{X^{l}_t\in [i_2/n_r,(i_2+1)/n_r)} , \\
      &\text{ with } \widehat{\alpha}^{l_1,l}_{i_1,i_2}=\frac{\sum_{j=1}^{J_v}\widehat{SCR}^{Nested,j}_t\mathds{1}_{X^{l_1,j}_t\in [i_1/n_r,(i_1+1)/n_r)} \mathds{1}_{X^{l,j}_t\in [i_2/n_r,(i_2+1)/n_r)}  }{\sum_{j=1}^{J_v} \mathds{1}_{X^{l_1,j}_t\in [i_1/n_r,(i_1+1)/n_r)} \mathds{1}_{X^{l,j}_t\in [i_2/n_r,(i_2+1)/n_r)}}.
\end{align*}
We then select the regressor $l_2\in \{1,\dots,12\} \setminus \{l_1 \}$ that gives the smallest RMSE. We then proceed similarly to select the next variables. We have run this selection for $t=10$ with $n_r=5$. Table~\ref{BA} shows the result of this algorithm and indicate the Book values of bonds as the more significant variable to approximate the SCR module on interest rates. We notice that the RMSE is significantly reduced by using the second variable. In contrast, the third variable moderately improves the criterion. Since the number of variables is also a limitation then for the use of the LSMC estimator, we do not go further in the selection procedure. 
Figure~\ref{fig:regression} illustrates the approximation of the values of~$\widehat{SCR}^{Nested,j}_t$ by the regression function with the two first regressors.  
\begin{figure}[h!]
  \includegraphics[width=0.8\textwidth]{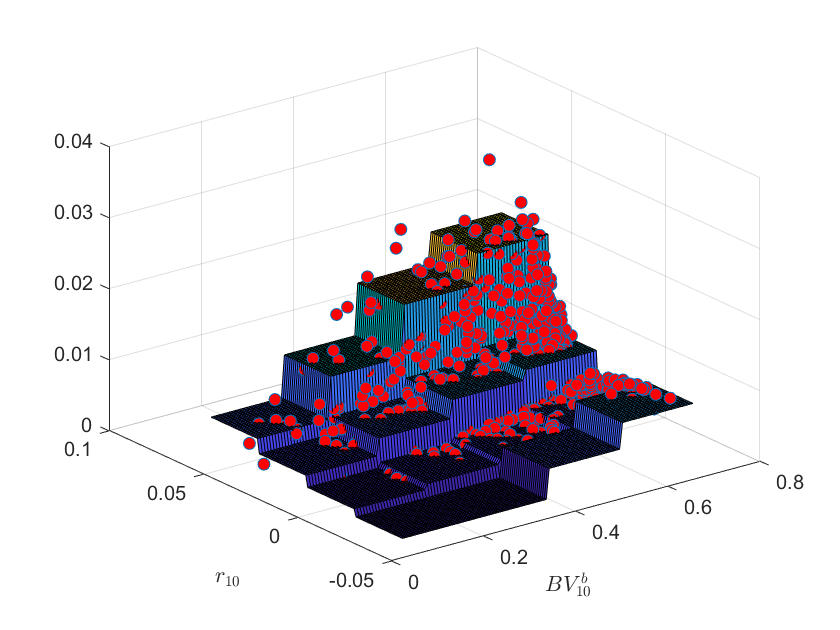}
  \caption{Plot of the points $\widehat{SCR}^{Nested,j}_t$ and of the estimated regression function with the two first selected regressors.}\label{fig:regression}
\end{figure}
\begin{table}[h]
\centering
\begin{tabular}{|l|l|l|}
\hline
 & Best attribute &   $RMSE$  \\\hline
First variable & $BV^b_t$ ($l_1=5$)  & 0.8739 \\\hline
Second variable & $r_t$ ($l_2=2$) &0.5292       \\\hline
Third variable & $S_t$ ($l_3=1$) &0.5084      \\\hline      
\end{tabular}
  \captionof{table}{Result of the Forward selection procedure for $SCR^{int}_t$ with $t=10$.}\label{BA}  
\end{table}


\subsubsection{Comparison between MLMC and LSMC estimators}\label{Comp_MLMC_LSMC}
We focus on the calculation of $\E[SCR^{int}_t]$ with $t=10$ years.  We now compare and test numerically the MLMC antithetic estimator $\widehat{I}^{MLMC}_A$ defined by~\eqref{Antithetic_MLMC_estim} with the LSMC estimator~$\widehat{I}^{LSMC}$ defined by~\eqref{def_ILSMC}, using the local cube basis and the regressors selected by the procedure described in Paragraph~\ref{parag_selec}.

In Figure~\ref{RMSE_different_estimator_10Y}, we have drawn the Root Mean Square Error of the estimator $\widehat{I}^{MLMC}_A$ and of the estimators $\widehat{I}^{LSMC}$ obtained by using the first, the two first and the three first selected regressors. In order to derive the RMSE of the different estimators, as no closed formulas is available in this framework, we rely on a full nested Monte-Carlo procedure based on a fixed simulation budget of $\Gamma=10^8$ sample paths to approximate the true value of~$I$. The allocation between primary and secondary scenarios correspond  to $M\approx\Gamma^{\frac{2}{3}}$ primary samples and $K\approx\Gamma^{\frac{1}{3}}$ inner scenarios, as prescribed by Theorem~\ref{thm_nested} for $\eta=1$. To compute the RMSE of the different estimators, we produce $N_{batch}=10$ independent simulations $(\widehat{I}^*_j)_{j=1,\ldots,N_{batch}}$  and indicate the empirical RMSE $\sqrt{\frac{1}{N_{batch}}\sum_{j=1}^{N_{batch}}\vert \widehat{I}^*_j-I\vert^2 }$. We plot the empirical RMSE's of the different estimators as a function $J$ (with $J:=\sum_{l=0}^LJ_lK_l$ for the MLMC estimator). This represents the number of samples, as well as the computational cost (in log-scale) that is in $O(J)$ for both estimators.

\begin{figure}[h!]\centering
    \includegraphics[width=0.7\textwidth]{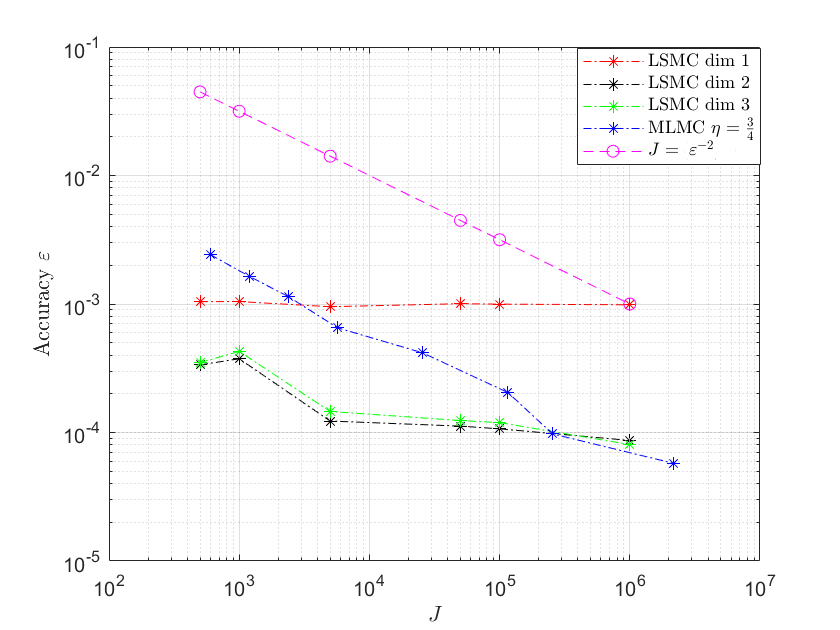}
    \caption{Empirical RMSE's of the LSMC and MLMC estimators in function of the computational effort $J$ (with $J:=\sum_{l=0}^LJ_lK_l$ for the MLMC estimator). The computational time needed for the forward selection used by the LSMC estimator is not taken into account in this plot.  }
    \label{RMSE_different_estimator_10Y}  
\end{figure}

Concerning the LSMC estimators, we notice that the estimator with two regressors does much better than the estimator with one regressor. Instead, the interest of using a third regressor is tiny. We also observe that the RMSE does not really decrease after $10^4$ samples on our example. This is due to the regression error: since we approximate $SCR^{int}_t$ by a function of two or three variables, there is no way to go beyond a certain level of precision. This is particularly noticeable for $t\ge 10$ years: the projection of the balance sheet through the ALM model is truly non-Markovian and the history up to time~$t$ cannot be summarized by two or three variables. In comparison, the convergence of the MLMC antithetic estimator is in line with Theorem~\ref{thm_MLMCA} and is asymptotically more accurate than the LSMC estimator. Besides, the MLMC estimator avoids the step of selecting regressors that requires computational time and may be determinant for the accuracy of the LSMC estimator. Last, we notice that for a same level of precision, the computational time required by the MLMC is slightly smaller than the one required by the LSMC estimator. More precisely, the computational time needed for $J=2\times 10^5$ (where the three estimators have quite the same accuracy) are  9950 seconds for $\widehat{I}^{MLMC}_A$, 11230 seconds for $\widehat{I}^{LSMC}$ with two regressors $(n_r=23)$ and  12650 seconds for $\widehat{I}^{LSMC}$ with three regressors $(n_r=13)$.

\begin{figure}[h!]
\centering
    \includegraphics[width=0.5\textwidth]{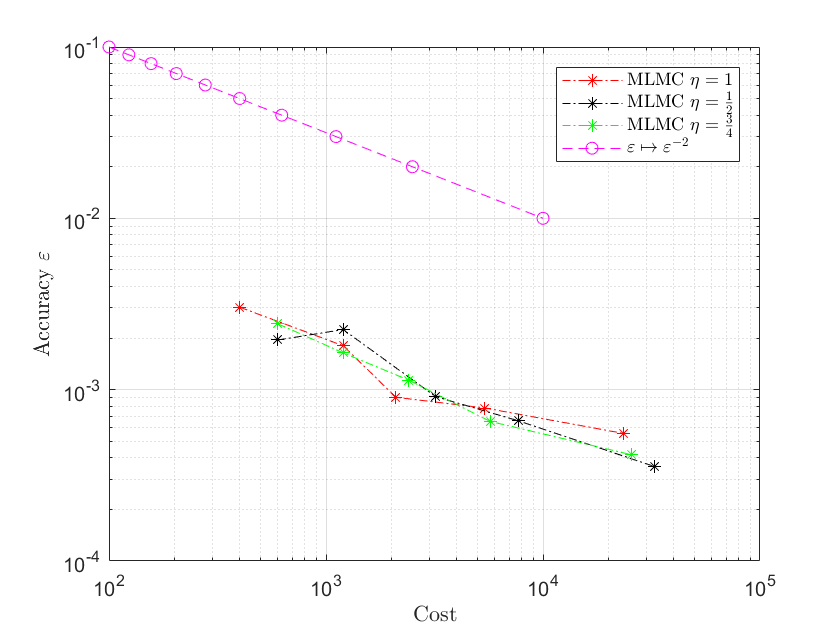}
    \caption{Convergence of the MLMC estimator for different values of $\eta$ in function of the computational cost $\sum_{l=0}^LJ_lK_l$.}
    \label{eta_ALM}  
\end{figure} 
We now make a comment on the choice of the parameter $\eta$ for the MLMC antithetic estimator. We recall that $\eta$ is, roughly speaking, related to the probability that two (or more) arguments of the maximum function are close to the maximum, see Theorems~\ref{thm_nested} and~\ref{thm_MLMCA}. Heuristically, the smaller is this probability, the larger can be $\eta$, which then reduces the number of levels and then the computational cost. In Figure~\ref{eta_ALM}, we have plotted the convergence of the MLMC antithetic estimator for $\eta \in \{1/2,3/4,1\}$ in function of the theoretical computational cost $\sum_{l=0}^L J_lK_l$. Basically, the three estimators converge, but the one obtained with $\eta=1$ does not seem to be asymptotically in $O(\varepsilon^{-2})$ while the two others are in line with the theoretical convergence in $O(\varepsilon^{-2})$. This shows the interest of the parameter~$\eta$ in a practical application and explains why we have chosen to take $\eta= 3/4$ is our experiments in Figure~\ref{RMSE_different_estimator_10Y}.

\begin{remark}
   In all our numerical experiments, we have used for the MLMC antithetic estimator the parameters given by Theorem~\ref{thm_MLMCA}. However, in practice, it may be wise when doing a MLMC to estimate the variance $\hat{V}_l$ of each level (with a few number of samples) and then to optimize $J_0,\dots,J_L$ with these values (i.e. to minimize the variance proxy of the estimator $\sum_{l=0}^L\hat{V}_l/J_l$ under the constraint $\sum_{l=0}^LJ_lK_l=C\varepsilon^{-2}$ on the computational cost, for some constant $C>0$). For sake of clarity and of reproducibility of our results, we do not have consider this empirical improvement of the MLMC antithetic estimator. 
\end{remark}

\subsubsection{Comparison between LSMC and the use of Neural Network}

In Paragraph~\ref{Comp_MLMC_LSMC}, we have noticed that a significant drawback of the LSMC estimator with respect to the MLMC antithetic estimator is that it requires first to select regressors. Beyond the computational time needed by this selection, there is a significant regression error.
\begin{table}[h!]
\centering
\begin{tabular}{|l|l|}
\hline
Input &Feature   \\\hline
$X_t^{1}$ &Bond Book-Value $BV^b_t$ \\\hline
$X_t^{2}$ &  Stock Book-value $BV^s_t$   \\\hline
$X_t^{3}$ & Position in bond $\phi^b_t$   \\\hline
$X_t^{4}$ & Position in stock $\phi^s_t$   \\\hline  
$X_t^{5}$ & Profit-Sharing Reserve level $PSR_t$   \\\hline   
$X_t^{6}$ & Mathematical Reserve level $MR_t$   \\\hline   
$X_t^{7}$ & Capitalization Reserve level $CR_t$   \\\hline   
$X_t^{8}$ & Spread crediting rate/competing rate  $\Delta_t$   \\\hline  
$X_t^{9}$ & Stock price  $S_t$   \\\hline  
$X_t^{10}$ & Interest-rate $r_t$   \\\hline 
\end{tabular}
  \captionof{table}{Inputs feature of the Neural Network}  
  \label{Input_NN}
\end{table}
A natural idea to skip the selection step in the regression is to use Neural Networks (NN). We have implemented a feedforward neural network with one hidden-layer. The hidden-layer is made with $10$ or $50$~neurons. The activation function used is the sigmoid function. To train the network, we generate $J$ outer scenarios $(X_{t,j})_{1\le j\le J}$ for which only one inner simulation $Z_j$ is performed and represents the maximal variation of the discounted P\&L due to the upward and downward shocks. Then, one minimizes
\begin{equation}\label{NN_tominimize}\frac 1 J \sum_{j=1}^J (Z_j-\text{NN}(X_{t,j}))^2,
\end{equation}
where $\text{NN}$ is the function generated by the neural network, so that it approximates the desired conditional expectation. The input features of the network have not been pre-processed. More precisely, the Neural Networks approximate a function on the 10-dimensional space defined by the inputs of Table~\ref{Input_NN}. The optimization above has to be enough for detecting the relevant variables for the approximation of the conditional expectation. Once the NN has been obtained, we then estimate $\E[SCR^{int}_t]$ simply by the empirical mean $\hat{I}^{NN}:=\frac{1}{J} \sum_{j=1}^J \text{NN}(X_{t,j}) $. Since we only use the NN on the training sample, the standard problem of overfitting is not an issue for our application. We compare the RMSE of this estimator with the RMSE obtained with the LSMC method. The aim of our procedure is first to assess if a  NN with a whole range of input features is able to select relevant attributes and second to compare with the LSMC method with well-chosen features.

\begin{table}[h!]
\hspace*{1cm}
\begin{minipage}[b]{.99\textwidth}
   \centering
 \begin{tabular}{|l|l|l|l|l|l|}
\hline 
 $J$ & LSMC dim 1& LSMC dim 2& LSMC dim 3& NN: 10 neurons& NN: 50 neurons \\\hline
500        & 1.0e-3& 3.36e-4&3.50e-4&7.075e-4&7.56e-4         \\\hline
$10^3$        & 1.0e-3& 3.75e-4&4.28e-4&6.46e-4&3.017e-4 \\\hline
$5\times10^3$        & 9.52e-4& 1.23e-4&1.46e-4&1.63e-4&1.8153e-4 \\\hline
$5\times10^4$        & 1.0e-3& 1.12e-4&1.24e-4&6.60e-5&7.29e-5 \\\hline
$10^5$        & 9.97e-4& 1.068e-4&1.19e-4&6.32e-5&6.92e-5 \\\hline
$10^6$        & 9.86e-4& 8.67e-5&8.06e-5&4.22e-5&4.50e-5         \\\hline
\end{tabular} 
 \captionof{table}{ RMSE of $\E[SCR^{int}_t]$ for $t=10$ given by the Neural Network (one hidden layer with the indicated number of neurons) and the LSMC in function of~$J$}   \label{cost_d1}
\end{minipage}
\end{table}

Table~\ref{cost_d1} indicates the RMSE of the estimator with the different methods. First, we notice that there is no need on our example to consider many neurons: a simple layer with $10$ neurons in enough and do as well as the NN with 50 neurons in terms of RMSE. We notice also that the estimator given by the NN is slightly better than the one obtained with the LSMC with two or three regressors when the training sample gets large. However, the use of neural networks present serious drawbacks. First, it requires to store all the samples to achieve the minimization of~\eqref{NN_tominimize} while the LSMC (and also MLMC) estimator only uses once each sample. Second, the time needed by the minimization (indicated in Table~\ref{cost_training}) is important, making at the end this method less competitive than MLMC. Note that one could be then tempted to train the NN on a smaller size of samples and then use it for large~$J$: one would then face the problem of overfitting, which we want to avoid.
 \begin{table}[h!]
\hspace*{1cm}
\begin{minipage}[b]{.99\textwidth}
   \centering
 \begin{tabular}{|l|l|l|l|l|l|}
\hline 
$J$  & $500$ & $10^3$ & $5\times 10^3$ & $5\times 10^4$  &$ 10^5$  \\ \hline
Time (s)&  18.8 & 30.4 & 50.5 & 273.1 &1693 \\ \hline
 \end{tabular}
 \captionof{table}{Time required in seconds for the optimization of~\eqref{NN_tominimize} for a neural network with 10 neurons.}  
 \label{cost_training}
\end{minipage}
 \end{table}
 
To sum up, the use of NN can indeed be useful to reduce the approximation error observed by using LSMC estimators. However, it both demands memory and computational time, making the gain with respect to the LSMC not obvious. The MLMC estimator presents the clear advantage to avoid this issue of function approximation, and to avoid any storage of data.

\subsection{Numerical experiments II: some insights on the ALM}\label{subsec_NUMII}

We now present some applications of the MLMC antithetic estimator for the ALM. One of the major issue in ALM is to determine the optimal asset allocation between the different asset class backed to the insurance portfolio. For that reason, it is crucial to evaluate precisely the amount of SCR required by the strategy. 
We are interested in calculating $\E[SCR^{mkt}_t]$, $\E[SCR^{eq}_t]$ and  $\E[SCR^{int}_t]$, see Subsection~\ref{subsec_STD} for the definition of these modules.  Since $SCR^{eq}_t$ and $SCR^{int}_t$ are random variables and the aggregation formula~\eqref{def_SCR_mkt} is not linear,  $\E[SCR^{mkt}_t]$ cannot be obtained by applying this formula to $\E[SCR^{eq}_t]$ and $\E[SCR^{int}_t]$. Note that by using the MLMC Antithetic estimator, it is possible and easy to calculate at the same time all these expectations, see Remark~\ref{Rk_Anti_gen} for the general expression of these estimators that we use for different functions $h$. At each level~$l$, one simulates $J_l$ primary scenarios up to time~$t$. Then, one simulates for each primary scenario $K_l$ secondary scenarios, on which we perform four different evolutions: the first one without any shock, the second one with the equity shock, the third one with the upward shock on interest rates and the last one with the downward shock on interest rates. Then, one computes the corresponding empirical means related to the calculation of $SCR^{eq}_t$ and $SCR^{int}_t$ and use the aggregation formula~\eqref{def_SCR_mkt} for $SCR^{mkt}_t$.  Let us note that the discontinuity induced by the coefficient~$\varepsilon$ may in principle deteriorate the MLMC estimation. However, we have thus run the MLMC estimator with a regularization of this coefficient and we have noticed a tiny impact of the regularization. This can be heuristically understood from Figure~\ref{fig:activation}: the activation of $\varepsilon$ may occur on a wide range (perhaps the whole range) of values of $SCR^{int}$, which smooths the phenomenon. 

\begin{figure}[h!]
 \begin{minipage}[t]{0.48\textwidth}
   \includegraphics[width=\textwidth]{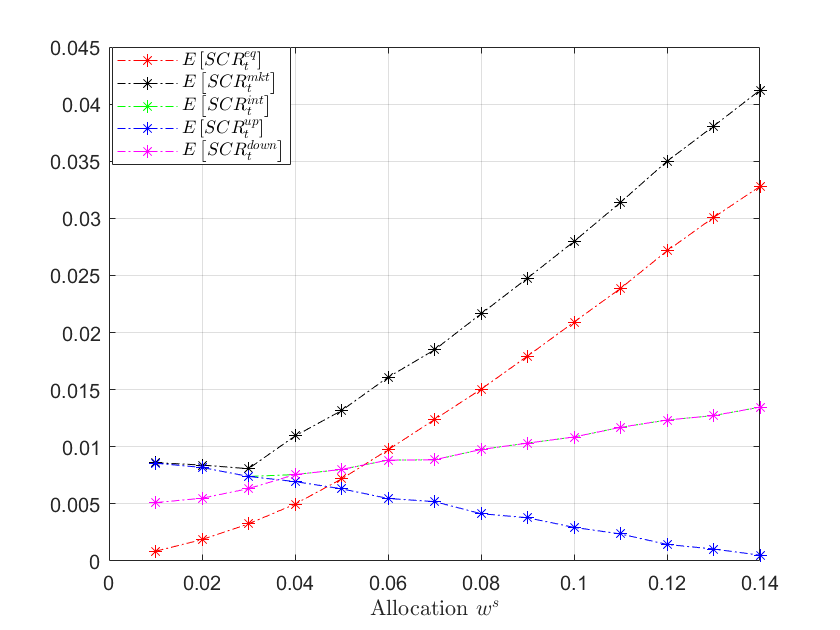}
\caption{Values of the SCR modules in function of the constant allocation weight $w^S$ in equity for $t=0$. } \label{SCR0}
 \end{minipage}
 \ 
 \begin{minipage}[t]{0.48\textwidth}
    \includegraphics[width=\textwidth]{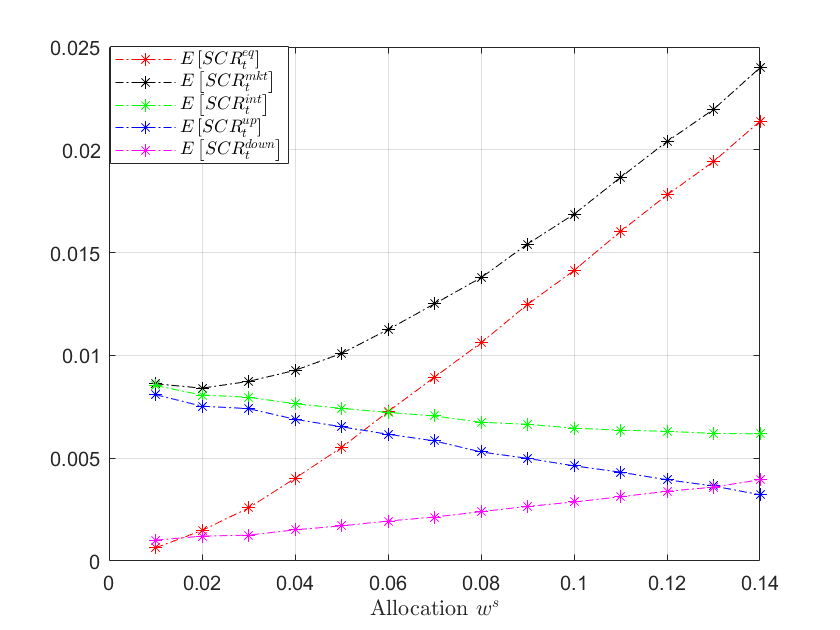}
    \caption{ Expected values of the SCR modules in function of the constant allocation weight $w^S$ in equity for $t=10$. }\label{SCR10}
 \end{minipage}
 \begin{minipage}[t]{0.48\textwidth}
   \includegraphics[width=\textwidth]{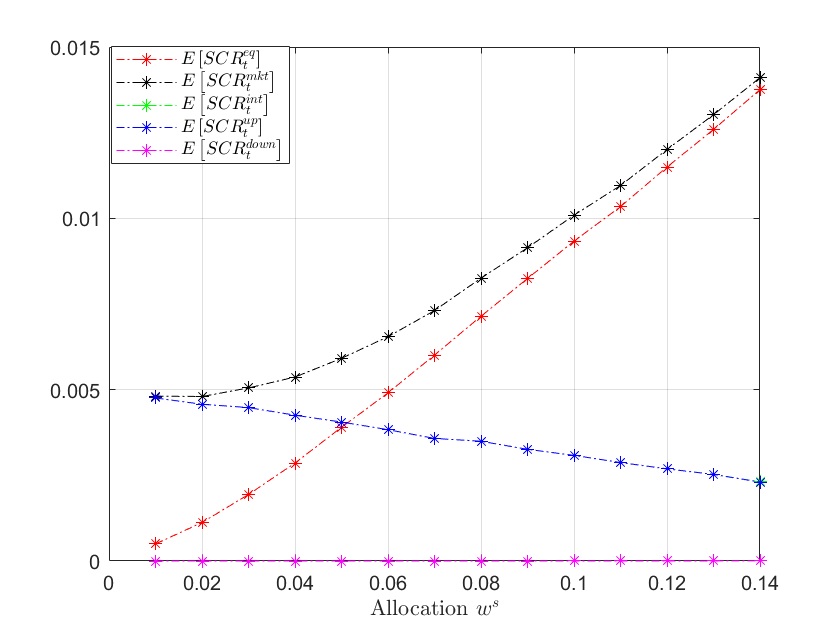}
\caption{Expected values of the SCR modules in function of the constant allocation weight $w^S$ in equity for $t=20$.} \label{SCR20}
 \end{minipage}
 \ \
 \begin{minipage}[t]{0.48\textwidth}
    \includegraphics[width=\textwidth]{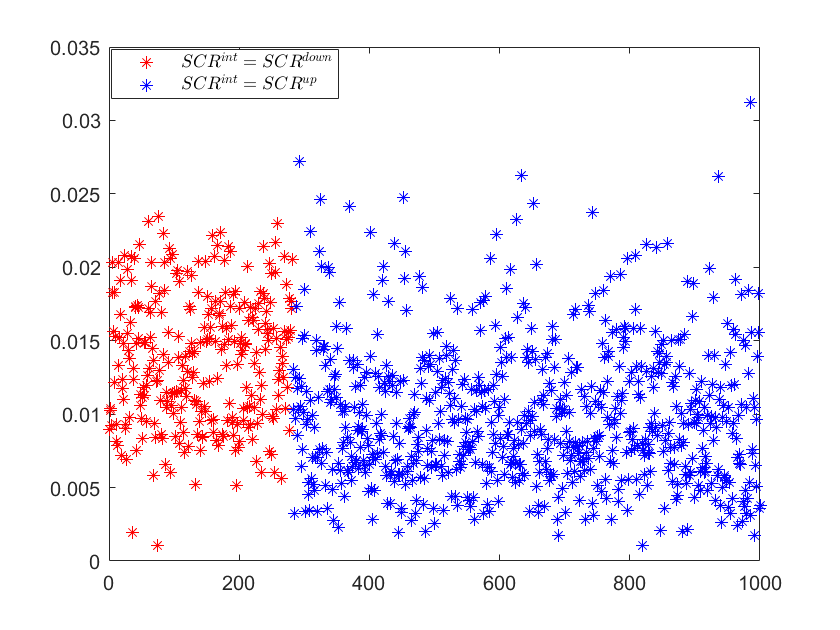}
    \caption{Sample of $1000$ approximated values of $SCR^{int}_{10}$ given by a nested estimator: in red (resp. blue) are the values obtained when the downward shock is greater (resp. smaller) than the upward shock on interest rates.}\label{fig:activation}
 \end{minipage}
\end{figure}

Figures~\ref{SCR0}, \ref{SCR10} and \ref{SCR20} illustrate respectively the different values of the SCR modules $\E[SCR^{mod}_t]$ for  $mod\in \{mkt, eq, int, up, down\}$ in function of the constant allocation weight $w^S$ in equity for $t=0$, $t=10$ and $t=20$ years (at $t=0$, we can remove the expectation). As one may expect, these values globally decrease with respect to the time since we are considering a run-off portfolio with an exit rate greater than $5\%$. We notice several interesting points. 
\begin{itemize}
\item At time $t=10$, the value of  $\E[SCR^{int}_t]$, which is equal to $\E[\max(SCR^{up}_t,SCR^{down}_t)]$,  is significantly larger than $\max(\E[SCR^{up}_t],\E[SCR^{down}_t])$. This shows that deterministic proxy values of SCR modules may induce errors. We no longer observe this phenomenon at time $t=20$ because the greater shock is always given by the upward shock and we have then $SCR^{int}_{20}\approx SCR^{up}_{20}$ (the green and blue curves coincide). This is explained in the next point.
\item The main effects of the shocks on the interest rate are the following. The upward (resp.~downward) shock leads to an immediate decrease (resp.~increase) of the portfolio market value, but on the long run higher (resp.~lower) rates gives a better (resp. worse) profitability. Here, we are considering a run-off portfolio with final maturity $T=30$. Thus, as $t$ increases, the effect on the long run of this shocks get less important making the immediate effect on market value dominant. Thus, at $t=20$ the downward shock is harmless while the upward shock gets painful. This explains why we observe  $SCR^{int}_{20}\approx SCR^{up}_{20}$.
\item The aggregation formula~\eqref{def_SCR_mkt} somehow encourages to have $SCR^{int}$ and $SCR^{eq}$ of the same order: if $SCR^{int}>>SCR^{eq}$ it is possible to invest more in equity to have a better average return with a moderate increase of $SCR^{mkt}$, and if $SCR^{int}<<SCR^{eq}$, one should reduce the investment in equity to reduce $SCR^{mkt}$. Therefore, it is interesting to look at the allocation that is such that $\E[SCR^{int}_t]=\E[SCR^{eq}_t]$. We see from Figures~\ref{SCR0}, \ref{SCR10} and \ref{SCR20} that the corresponding weight $w^S$ evolves slightly. We get $w^S\approx 0.05$ for $t=0$, $w^S\approx 0.06$ for $t=10$ and $w^S\approx 0.05$ for $t=20$, which is still a relative variation of~20\%. This shows that a better evaluation of the SCR along the time may lead to significant adjustments on the investment strategy.  
\end{itemize}

Besides the calculation of the SCR, we can also use of the MLMC Antithetic estimator to calculate the sensitivity of the SCR with respect to some variations of parameters or market prices. These sensitivities are interesting information for management and they can also be useful to calculate quickly a proxy value of the SCR. For example, suppose that we have computed the value of $\E[SCR^{mkt}_{10}]$ and, after few days, we want to estimate the new value of  $\E[SCR^{mkt}_{10}]$ by taking into account the small variation on the equity and on the interest rates. Then, we can do this easily if one has computed the values of the sensitivities $\frac{\E[SCR^{mkt}_{10}(r_0+\delta r_0)- SCR^{mkt}_{10}(r_0)] }{\delta r_0}$ and $\frac{\E[SCR^{mkt}_{10}(S_0+\delta S_0)- SCR^{mkt}_{10}(S_0)] }{\delta S_0}$, where implicitly all values are kept constant but respectively $r_0$ and $S_0$. Note that these sensitivities can be computed with the MLMC antithetic estimator with the same samples that are needed for the estimation of~$SCR^{mkt}$. Table~\eqref{sensi} indicates the sensitivities obtained with our parameters with $\delta r_0=0.001$ and $\delta S_0 = 0.01$.
 \begin{table}[h!]
\centering
\begin{minipage}[b]{0.8\textwidth}
   \centering
 \begin{tabular}{|c|c|}
\hline 
 $\frac{\E[SCR^{mkt}_{10}(S_0+\delta S_0)- SCR^{mkt}_{10}(S_0)] }{\delta S_0}$& $\frac{\E[SCR^{mkt}_{10}(r_0+\delta r_0)- SCR^{mkt}_{10}(r_0)] }{\delta r_0}$  \\\hline
          0.0228 & -0.0845         \\\hline
\end{tabular} 
 \captionof{table}{Sensitivities of the $SCR^{mkt}_{10}$ with $w^S=0.05$.}  
 \label{sensi}
\end{minipage}
 \end{table}

 Last, the MLMC estimation is a tool for example to analyse how the $SCR$ depends on the risk premia of stocks and interest rates. If the evaluation of the SCR has to be performed for regulatory reasons under a risk neutral framework, it is more relevant for ALM to calculate $\E^{\Pb}[SCR^{mod}_t]$ under the real probability, which corresponds to the average value of the own funds that will be necessary at time~$t$. From equations~\eqref{SCR_int_histo} (generalized to any other SCR module) and~\eqref{Chg_prob}, it is possible to see the impact of the risk premia $\lambda^W$ and $\lambda^Z$ on each SCR module. In Figures~\ref{SCR_int_lambdaZ} and~\ref{SCR_eq_lambdaW}, we have indicated the more  remarkable ones: the dependence of  $\E^{\Pb}[SCR^{int}_t]$ on $\lambda^Z$ and of $\E^{\Pb}[SCR^{eq}_t]$ on $\lambda^W$.  We notice that the larger is $\lambda^Z$, the larger is $\E^{\Pb}[SCR^{int}_t]$. This can be understood as follows. A higher $\lambda^Z$ leads to a higher mean reverting level for the short rate $r$. Thus, under the real probability measure, bonds are better remunerated on the time interval $[0,t]$, and at time~$t$ the amount of savings (mathematical reserve) is higher. Since the evaluation of $SCR^{int}_t$ is risk neutral, $\lambda^Z$ has then no incidence on this evaluation. Thus, we observe a higher value of $\E^{\Pb}[SCR^{int}_t]$ simply because the mathematical reserve at time~$t$ is higher because of better returns.  The same interpretation holds for $\E^{\Pb}[SCR^{eq}_t]$: the higher is $\lambda^W$, the higher is the amount of savings at time~$t$ and therefore the higher is  $\E^{\Pb}[SCR^{eq}_t]$. Last, let us mention that for the simulations of Figures~\ref{SCR_int_lambdaZ} and~\ref{SCR_eq_lambdaW}, we have run independently the MLMC algorithm for each time~$t$. A natural question then is if we could handle all these calculations together to spare some computational time. This is left for further research.

\begin{figure}[h!]
  \begin{minipage}[t]{0.48\textwidth}
   \includegraphics[width=\textwidth]{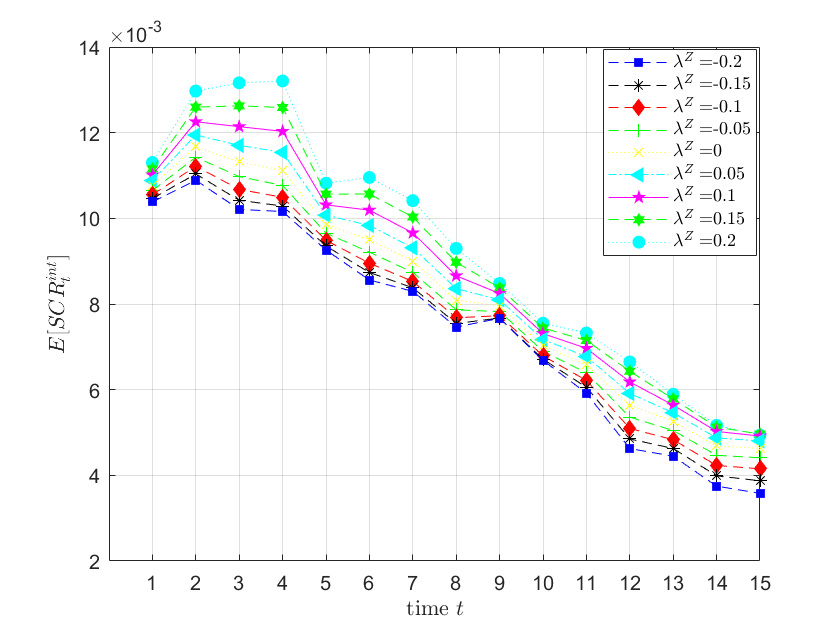}
\caption{Estimated values of~$\E^{\Pb}[SCR^{int}_t]$ in function of $t$ for different risk premia $\lambda^Z$.} 
    \label{SCR_int_lambdaZ}
  \end{minipage}
  \ 
   \begin{minipage}[t]{0.48\textwidth}
   \includegraphics[width=\textwidth]{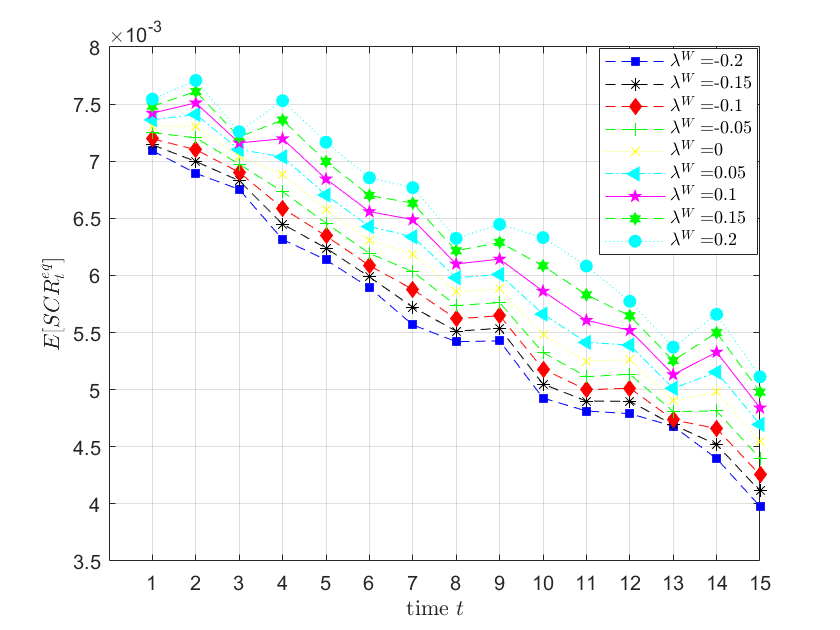}
\caption{Estimated values of~$\E^{\Pb}[SCR^{eq}_t]$ in function of $t$ for different risk premia $\lambda^W$. } \label{SCR_eq_lambdaW}
  \end{minipage} 
\end{figure}

\newpage
\appendix

\section{Technical proofs for Theorems~\ref{thm_nested} and~\ref{thm_MLMCA} }\label{Appendix_MLMC}

\subsection{Preliminary results}
In this section, we gather elementary but useful results for the analysis of the nested and the multilevel Monte-Carlo estimators.

\begin{proposition}
  Let $g_0(u)=\max \{u,0\}$ and, for $\varepsilon>0$, $g_\varepsilon(u)=\frac{u^2}{2\varepsilon}\mathds{1}_{u\in[0,\varepsilon]}+(u-\frac{\varepsilon}{2})\mathds{1}_{u>\varepsilon}$. The function $g_\varepsilon$ is $C^1$ and piecewise $C^2$ with 
\begin{align*}
g^{\prime}_\varepsilon(u)&=\frac{u}{\varepsilon}\mathds{1}_{u\in[0,\varepsilon]}+\mathds{1}_{u>\varepsilon}, \ g^{\prime\prime}_\varepsilon(u)=\frac{1}{\varepsilon}\mathds{1}_{u\in[0,\varepsilon]}
\end{align*}
Moreover, $g_\varepsilon$ is 1-Lipschitz and we have
\begin{align*}
g_\varepsilon\leq g_0\leq g_\varepsilon+\frac{\varepsilon}{2}.
\end{align*}
In addition, for any $\theta,\widehat{\theta}$ such that $\theta\leq\widehat{\theta}\in \mathbb{R}$ we have: 
\begin{equation}
\forall a \in \R,  0\leq\int_{\theta}^{\widehat{\theta}}g^{\prime\prime}_\varepsilon(t+a)dt\leq \int_\mathbb{R}g^{\prime\prime}_\varepsilon(t)dt=1.
\label{int_g_eps_1}
\end{equation}
Finally, the following asymptotic properties holds : $\forall\ u\in \mathbb{R}\ g_\varepsilon(u)\xrightarrow[\varepsilon \to 0]{}g_0(u)$,  $\forall\ u\in \mathbb{R}\ g^{\prime}_\varepsilon(u)\xrightarrow[\varepsilon \to 0]{}\mathds{1}_{u>0}  $
and $ \int_\mathbb{R}g_\varepsilon^{\prime\prime}(u)\varphi(u) du\xrightarrow[\varepsilon \to 0]{}\varphi(0) $ for any function $\varphi:\R \to \R$ that is right-continuous at~$0$. 
\label{g_eps_limit}
\end{proposition}
\begin{lemma}
Let $ \theta,\widehat{\theta}\in \mathbb{R}$ and $\left\{a_\varepsilon\right\}$ an arbitrary function that converges to $a$ as $\varepsilon\rightarrow 0$, then: 
\begin{equation}
\limsup_{\varepsilon\rightarrow 0}\left\vert \int_{\theta}^{\widehat{\theta}}g^{\prime\prime}_\varepsilon(t+a_\varepsilon)dt\right\vert\leq \left\vert \mathds{1}_{\theta\leq-a\leq \widehat{\theta}}-\mathds{1}_{\widehat{\theta}\leq-a\leq\theta}\right\vert
\end{equation}
 \label{lemma:integral_limit_eps}
\end{lemma}
\begin{proof}[Proof of \Cref{lemma:integral_limit_eps}]
 Without loss of generality, we assume that $\theta<\widehat{\theta}$. 
 First we have that :
 \begin{align*}
\int_\theta^{\widehat{\theta}}g^{\prime\prime}_\varepsilon(t+a_\varepsilon)dt=\frac{1}{\varepsilon}\int_{A_\varepsilon}1 dt
\end{align*}
where $A_\varepsilon=[\theta,\widehat{\theta}]\cap[-a_\varepsilon,\varepsilon-a_\varepsilon]$. Hence, if $-a<\theta$ or $\widehat{\theta}<-a$, it exists $\varepsilon_0>0$ such that $\forall\varepsilon\in[0,\varepsilon_0], A_\varepsilon=\emptyset$. In this case: 
\begin{align*}
\limsup_{\varepsilon\rightarrow 0}\left\vert \int_{\theta}^{\widehat{\theta}}g^{\prime\prime}_\varepsilon(t+a_\varepsilon)dt\right\vert=0
\end{align*}
Otherwise (i.e if $\theta\leq -a\leq \widehat{\theta}$) we always have: $A_\varepsilon\subset[-a_\varepsilon,\varepsilon-a_\varepsilon] $. Therefore, we have
$\limsup_{\varepsilon\rightarrow 0}\left\vert \int_{\theta}^{\widehat{\theta}}g^{\prime\prime}_\varepsilon(t+a_\varepsilon)dt\right\vert\leq  1$ by~(\ref{int_g_eps_1}). Thus we obtain 
 \begin{align*}
 \limsup_{\varepsilon\rightarrow 0} \left\vert \int_{\theta}^{\widehat{\theta}}g^{\prime\prime}_\varepsilon(t+a_\varepsilon)dt\right\vert \leq \mathds{1}_{\left\{\theta\leq-a\leq \widehat{\theta}\right\}}. \quad  \qedhere
 \end{align*}
\end{proof}

\begin{lemma}
\label{lemma::indicatrice_cv_markov}
Let $\theta,\widehat{\theta}\in \R$. Then, we have
\begin{align}
& \left\vert \mathds{1}_{\theta \leq 0\leq \widehat{\theta}}-\mathds{1}_{\widehat{\theta}\leq 0 \leq \theta}\right\vert \le
\mathds{1}_{\theta\widehat{\theta} \leq 0},
\label{ecart_indic} \\
& \mathds{1}_{\theta\widehat{\theta}\leq 0}\leq\frac{\vert\widehat{\theta}-\theta\vert}{\vert\theta\vert} \text{ for } \theta \not = 0.
\label{indic_bound}
\end{align}
\end{lemma}
\begin{proof}
  Inequality~\eqref{ecart_indic} is an equality when $\theta\not=0$ or $\theta=0,\widehat{\theta}\not = 0$ and an obvious inequality when  $\theta=\widehat{\theta} = 0$. 
  The right hand side of~\eqref{indic_bound} is nonnegative. When $\theta<0$ and $\widehat{\theta}\ge 0$ (resp. $\theta>0$ and $\widehat{\theta}\le 0$) , we have $|\widehat{\theta} - \theta|=\widehat{\theta} - \theta\ge -\theta =|\theta|$ (resp. $|\widehat{\theta} - \theta|=-\widehat{\theta} + \theta\ge \theta =|\theta|$).  
\end{proof}


\begin{lemma}\label{lem_speedLLN}
  Let $(Z_k)_{k\in \N^*}$ be an i.i.d. sequence of square integrable real valued random variables. Let $\mu=\E[Z_1]$ and $\sigma=\sqrt{\V[Z_1]}$. For $\gamma \in [1,2]$, we define $D_\gamma=\sigma^\gamma$ and for $\gamma>2$, $D_\gamma=\E[|Z_1-\mu|^\gamma] \in [0,+\infty]$. Then, we have 
  $$ \E\left[\left|\frac 1 K \sum_{k=1}^K Z_k -\mu \right|^\gamma\right]\le C_\gamma \frac{D_\gamma}{K^{\gamma/2}},$$
with $C_\gamma=1$ for $\gamma\in [1,2]$ and $C_\gamma=\left(2\sqrt{\gamma-1}\right)^\gamma$ for $\gamma>2$.
\end{lemma}
\begin{proof}
For $\gamma\in [1,2]$, We have from Jensen inequality $\E\left[\left|\frac 1 K \sum_{k=1}^K Z_k -\mu \right|^\gamma\right]\le\E\left[\left|\frac 1 K \sum_{k=1}^K Z_k -\mu \right|^2\right]^{\gamma/2}=\sigma^\gamma/K^{\gamma/2}$. For $\gamma>2$, this result is stated in Corollary 2.5~\cite{GHJVW}
\end{proof}

\subsection{Nested Monte-Carlo estimator}\label{subsec_app_nested}

\begin{proof}[Proof of Lemma~\ref{lemma_rec_max}]
Let $\bar{b}(X)=\mathbb{E}\left[\max\{\widehat{\theta}_K^{1},\widehat{\theta}_K^{2}\}-\max\{\varphi_1(X),\varphi_2(X)\}|X\right]$. Since  $ \max\{a,b\}=a+g_0(b-a)$ for $a,b\in\R$, we get 
 \begin{equation}
 \bar{b}(X)=\mathbb{E}\left[\widehat{\theta}_K^{1}-\varphi_1(X)+g_0(\widehat{\theta}_K^{21})-g_0(\varphi_{21}(X))|X\right]
 \end{equation}
 Now, we observe that
 \begin{equation}\label{smooth_b}
   \bar{b}(X)= \lim_{\varepsilon\rightarrow 0} \bar{b}_\varepsilon(X) \text{ with } \bar{b}_\varepsilon(X)=\mathbb{E}\left[\widehat{\theta}_K^{1}-\varphi_1(X)+g_\varepsilon(\widehat{\theta}_K^{21})-g_\varepsilon(\varphi_{21}(X))|X\right].
 \end{equation}
 Since $g_\varepsilon$ is 1-Lipschitz, we have
\begin{align*}
\vert g_\varepsilon(\widehat{\theta}_K^{21})-g_\varepsilon(\varphi_{21}(X))\vert\leq \vert\widehat{\theta}_K^{21}-\varphi_{21}(X)\vert\leq \vert\widehat{\theta}_K^{1}-\varphi_{1}(X)\vert+\vert\widehat{\theta}_K^{2}-\varphi_{2}(X)\vert.
\end{align*}
Then, we get~\eqref{smooth_b} by using the integrability assumption~(ii) and Lebesgue's dominated convergence theorem.
Since $g_\varepsilon$ is $C^1$ and piecewise $C^2$, we can make a Taylor expansion to obtain: 
\begin{equation*}
\bar{b}_\varepsilon(X)=\mathbb{E}\left[\widehat{\theta}_K^{1}-\varphi_1(X)+g_\varepsilon^{\prime}(\varphi_{21}(X))\left(\widehat{\theta}_K^{21}-\varphi_{21}(X)\right)+\int_{\varphi_{21}(X)}^{\widehat{\theta}_K^{21}}\left(\widehat{\theta}_K^{21}-t\right)g_\varepsilon^{\prime\prime}(t)dt|X\right]
\end{equation*}
 Then, since $g_\varepsilon^{\prime}(\varphi_{21}(X))$ is $\sigma(X)$-measurable, using Proposition \ref{g_eps_limit} we get
 \begin{align*}
   \lim_{\varepsilon\rightarrow 0}\mathbb{E}\left[\widehat{\theta}_K^{1}-\varphi_1(X)+g_\varepsilon^{\prime}(\varphi_{21}(X))\left(\widehat{\theta}_K^{21}-\varphi_{21}(X)\right)|X\right]&=\mathbb{E}\left[\widehat{\theta}_K^{1}-\varphi_1(X)+\mathds{1}_{\varphi_{21}(X)>0}\left(\widehat{\theta}_K^{21}-\varphi_{21}(X)\right)|X\right] \\
   &=\mathbb{E}\left[\mathds{1}_{\varphi_{21}(X)\le0}\left(\widehat{\theta}_K^{1}-\varphi_1(X)\right)+\mathds{1}_{\varphi_{21}(X)>0}\left(\widehat{\theta}_K^{2}-\varphi_{2}(X)\right)|X\right]
 \end{align*}
 Using condition~(ii) we get :
 \begin{equation}
\left\vert\mathbb{E}\left[\mathds{1}_{\varphi_{21}(X)\le0}\left(\widehat{\theta}_K^{1}-\varphi_1(X)\right)+\mathds{1}_{\varphi_{21}(X)>0}\left(\widehat{\theta}_K^{2}-\varphi_{2}(X)\right) |X\right]\right\vert\leq \frac{\mathds{1}_{\varphi_{21}(X)\le0} C_1(X)+\mathds{1}_{\varphi_{21}(X)>0}C_2(X)}{K^{\frac{1+\eta}2}} \label{biais_1A}
 \end{equation}
 Now we focus on the remainder in the Taylor decomposition. Using Lemma~\ref{lemma:integral_limit_eps} and the dominated convergence theorem and then Lemma~\ref{lemma::indicatrice_cv_markov} with $\varphi_{21}(X)\not=0$ a.s. (Assumption (iii)), we get
\begin{align*}
\lim_{\varepsilon\rightarrow 0}\left\vert\mathbb{E}\left[\int_{\varphi_{21}(X)}^{\widehat{\theta}_K^{21}}\left(\widehat{\theta}_K^{21}-t\right)g_\varepsilon^{\prime\prime}(t)dt|X\right]\right\vert&\leq \mathbb{E}\left[\vert\widehat{\theta}_K^{21}-\varphi_{21}(X)\vert\limsup_{\varepsilon\rightarrow 0}\left| \int_{\varphi_{21}(X)}^{\widehat{\theta}_K^{21}}g_\varepsilon^{\prime\prime}(t)dt\right| |X\right]\\
&\leq \mathbb{E}\left[\vert\widehat{\theta}_K^{21}-\varphi_{21}(X)\vert\vert\mathds{1}_{\varphi_{21}(X)\leq 0\leq\widehat{\theta}_K^{21}}-\mathds{1}_{\widehat{\theta}_K^{21}\leq 0\leq \varphi_{21}(X)}\vert|X\right]\\
&\leq  \mathbb{E}\left[\vert\widehat{\theta}_K^{21}-\varphi_{21}(X)\vert\mathds{1}_{\widehat{\theta}_K^{21}\varphi_{21}(X)\leq 0}|X\right]\\ 
&\leq \mathbb{E}\left[\frac{\vert\widehat{\theta}_K^{21}-\varphi_{21}(X)\vert^{1+\eta}}{\vert\varphi_{21}(X)\vert^\eta}|X\right]=\frac{\mathbb{E}\left[\vert\widehat{\theta}_K^{21}-\varphi_{21}(X)\vert^{1+\eta}|X\right]}{\vert\varphi_{21}(X)\vert^\eta}
\end{align*}
Now, we use the norm inequality $\mathbb{E}\left[\vert\widehat{\theta}_K^{21}-\varphi_{21}(X)\vert^{1+\eta}|X\right]^{\frac 1 {1+\eta}}\le \mathbb{E}\left[\vert\widehat{\theta}_K^{1}-\varphi_{1}(X)\vert^{1+\eta}|X\right]^{\frac 1 {1+\eta}} + \mathbb{E}\left[\vert\widehat{\theta}_K^{2}-\varphi_{2}(X)\vert^{1+\eta}|X\right]^{\frac 1 {1+\eta}}$ and the convexity of $x\mapsto x^{1+\eta}$ to get
\begin{align}
  \mathbb{E}\left[\vert\widehat{\theta}_K^{21}-\varphi_{21}(X)\vert^{1+\eta}|X\right]&\leq 2^\eta
  \left( \mathbb{E}\left[\vert\widehat{\theta}_K^{1}-\varphi_{1}(X)\vert^{1+\eta}|X\right] + \mathbb{E}\left[\vert\widehat{\theta}_K^{2}-\varphi_{2}(X)\vert^{1+\eta}|X\right]  \right) \nonumber \\
  &
  \leq 2^\eta\frac{\sigma_1^{1+\eta}(X)+\sigma_2^{1+\eta}(X)}{K^{\frac{1+\eta}{2}}}
\label{strong_err12A}
\end{align}
With~\eqref{smooth_b}, equations~\eqref{biais_1A} and~\eqref{strong_err12A} give the bias estimate~\eqref{weak_err_max}.

We now focus on the variance. The proof is straightforward using condition (ii) and the inequality $|\max\{a_1,a_2\}-\max\{b_1,b_2\}|\le\max\{|a_1-b_1|,|a_2-b_2|\}$ : 
\begin{align*}
 \mathbb{E}\left[\left\vert\max\{\widehat{\theta}_K^{1},\widehat{\theta}_K^{2}\}-\max\{\varphi_1(X),\varphi_2(X)\}\right\vert^2|X\right]&\leq \mathbb{E}\left[\max\{\vert\widehat{\theta}_K^{1}-\varphi_1(X)\vert^2,\vert\widehat{\theta}_K^{2}-\varphi_2(X)\vert^2\}|X\right] \\
 &\leq  \mathbb{E}\left[\vert\widehat{\theta}_K^{1}-\varphi_1(X)\vert^2|X\right]+\mathbb{E}\left[\vert\widehat{\theta}_K^{2}-\varphi_2(X)\vert^2|X\right]\\
 &\le \frac{\sigma_1^2(X)+\sigma_2^2(X)}{K}= \frac{\sigma^2(X)}{K}.\qedhere
\end{align*}
\end{proof}

\subsection{Antithetic MLMC estimator}

In this section we prepare the proof of Theorem~\ref{thm_MLMCA} and start with a useful preliminary lemma.

\begin{lemma}
\label{induction_MLMC_anti}
Let $p\ge 2$ and $K\in 2\N^*$. With the notation introduced in~\eqref{Notation1} and~\eqref{Notation2}, the following property holds:
\begin{align*}
\E\left[\left\vert \widehat{M}^p_{K}-\frac{\widehat{M}^p_{K/2}+\widehat{M}^{p,\prime}_{K/2}}{2}\right\vert^2|X\right]& \leq 2 \E\left[\left\vert \widehat{M}^{p-1}_{K}-\frac{\widehat{M}^{p-1}_{K/2}+\widehat{M}^{p-1,\prime}_{K/2}}{2}\right\vert^2|X\right]\\
&+2\E\left[h^2\left(\widehat{M}^{p-1}_{K}-\widehat{E}^p_{K/2},\widehat{M}^{p-1,\prime}_{K/2}-\widehat{E}^{p,\prime}_{K/2}\right)|X\right],
\end{align*}
where $h(x,y)=\left(\frac{x+y}{2}\right)^+-\frac{(x)^++(y)^+}{2}$.
\end{lemma} 
\begin{proof}[Proof of Lemma \ref{induction_MLMC_anti}]
Observing that $\forall a,b\in \mathbb{R}$, $\max\{a,b\}=a+(b-a)^+$, we deduce that
\begin{align*}
\widehat{M}^p_{K}-\frac{\widehat{M}^p_{K/2}+\widehat{M}^{p,\prime}_{K/2}}{2}&=\max\{\widehat{E}^p_{K},\widehat{M}^{p-1}_{K}\}-\frac{\max\{\widehat{E}^p_{K/2},\widehat{M}^{p-1}_{K/2}\}+\max\{\widehat{E}^{p,\prime}_{K/2},\widehat{M}^{p-1,\prime}_{K/2}\}}{2} \\
&=\widehat{E}^p_{K}+(\widehat{M}^{p-1}_{K}-\widehat{E}^p_{K})^+-\frac{\widehat{E}^p_{K/2}+(\widehat{M}^{p-1}_{K/2}-\widehat{E}^p_{K/2})^++\widehat{E}^{p,\prime}_{K/2}+(\widehat{M}^{p-1,\prime}_{K/2}-\widehat{E}^{p,\prime}_{K/2})^+}{2}\\
&=(\widehat{M}^{p-1}_{K}-\widehat{E}^p_{K})^+-\frac{ (\widehat{M}^{p-1}_{K/2}-\widehat{E}^p_{K/2})^++(\widehat{M}^{p-1,\prime}_{K/2}-\widehat{E}^{p,\prime}_{K/2})^+}{2} \\
&=(\widehat{M}^{p-1}_{K}-\widehat{E}^p_{K})^+-\left(\frac{\widehat{M}^{p-1}_{K/2}+\widehat{M}^{p-1,\prime}_{K/2}}{2}-\widehat{E}^p_{K}\right)^++h\left(\widehat{M}^{p-1}_{K}-\widehat{E}^p_{K/2},\widehat{M}^{p-1,\prime}_{K/2}-\widehat{E}^{p,\prime}_{K/2}\right),
\end{align*}
using that $\widehat{E}^p_{K}=\frac{\widehat{E}^p_{K/2}+\widehat{E}^{p,\prime}_{K/2}}2$ for the third and fourth equality. We conclude using that $(a+b)^2\le 2(a^2+b^2)$. 
\end{proof}

The following lemmas gives a bound on $h$
\begin{lemma}
Let $h(x,y)=\left(\frac{x+y}{2}\right)^+-\frac{(x)^++(y)^+}{2}$ for $x,y\in\R$. Then, we have  
$h(x,y)=-\frac{\vert x\vert\wedge\vert y\vert}{2}\mathds{1}_{xy\leq 0}$.
 \label{lemma::h_function_bound_dim1}
\end{lemma}
\begin{proof} By distinguishing the cases as follows, we get the claim:
$$h(x,y)= \left\{
    \begin{array}{ll}
       y/2 & \mbox{if }x+y \ge 0,x>0,y<0,  \\
        x/2 & \mbox{if }  x+y \ge 0,x<0,y>0,\\
        -x/2 & \mbox{if }  x+y<0,x>0,y<0,\\
          -y/2 & \mbox{if }  x+y<0,x<0,y>0,\\
          0&\mbox{otherwise. } \qedhere
    \end{array}    \right.$$
\end{proof}

\begin{proposition}\label{prop_antithetic}
  We use the notation introduced in~\eqref{Notation1} and~\eqref{Notation2}. Let $\eta>0$ and $D^p_{2+\eta}(X)=\E[|Y^p-\E[Y^p|X]|^{2+\eta}|X]$. We assume that $\Pb(E^p_X=M^{p-1}_X)=0$ for all $p\in\{2,\dots,P\}$ and
  \begin{equation*}
    \forall p \in \{2,\dots,P\},\ \E\left[ \frac{D^p_{2+\eta}(X)}{|E^p_X-M^{p-1}_X|^\eta} \phi^2(X) \right]<\infty.
  \end{equation*}
 Then, there exist a constant $C\in \R_+^*$ such that
  $$\V\left[\left(\widehat{M}^P_{K}-\frac{\widehat{M}^P_{K/2}+\widehat{M}^{p,\prime}_{K/2}}{2}\right)\phi(X) \right]\le \frac{C}{K^{1+\eta/2}}.$$
\end{proposition}
\begin{proof}
  Let us define for $p=1,\dots,P$,
  \begin{align*}
    U^p_K&=\E\left[\left\vert \widehat{M}^p_{K}-\frac{\widehat{M}^p_{K/2}+\widehat{M}^{p,\prime}_{K/2}}{2}\right\vert^2\bigg|X\right] \\
  \varepsilon^p_K&=\E\left[h^2\left(\widehat{M}^{p-1}_{K}-\widehat{E}^p_{K/2},\widehat{M}^{p-1,\prime}_{K/2}-\widehat{E}^{p,\prime}_{K/2},\right)|X\right].
  \end{align*}
  We notice that $U^1_K=0$, and  Lemma \ref{induction_MLMC_anti} gives $U^p_K\le 2(U^{p-1}_K+\varepsilon^p_K)$ for $p=2,\dots,P$. A straightforward induction leads to
\begin{equation}
U^P_K \leq \sum_{p=2}^{P} 2^{P+1-p} \varepsilon^{p}_K.
\end{equation}
The variance being smaller than the expectation of the square, we get by using the tower property of the conditional expectation
\begin{equation}\label{majo_var}
  \V\left[\left(\widehat{M}^P_{K}-\frac{\widehat{M}^P_{K/2}+\widehat{M}^{p,\prime}_{K/2}}{2}\right)\phi(X) \right]\le  \sum_{p=2}^{P} 2^{P+1-p} \E[\varepsilon^{p}_K\phi^2(X)].
\end{equation}
For $p=2,\dots,P$, we define the following random variables 
$$H_X^p=M_X^{p-1}-E_X^p, \ \widehat{H}^p_{K/2}=\widehat{M}^{p-1}_{K/2}-\widehat{E}_{K/2}^p,\ \widehat{H}^{p,\prime}_{K/2}=\widehat{M}^{p-1,\prime}_K-\widehat{E}_K^{p,\prime}.$$
We now use Lemma~\ref{lemma::h_function_bound_dim1} and the equality $\mathds{1}_{\widehat{H}^{p}_{K/2}\widehat{H}^{p,\prime}_{K/2}<0}= \mathds{1}_{\widehat{H}^{p}_{K/2}\widehat{H}^{p,\prime}_{K/2}<0} \mathds{1}_{\widehat{H}^{p}_{K/2} H^p_X<0} +  \mathds{1}_{\widehat{H}^{p}_{K/2}\widehat{H}^{p,\prime}_{K/2}<0} \mathds{1}_{ \widehat{H}^{p,\prime}_{K/2} H^{p}_X<0}$ that is true a.s. since $\Pb(H_X^p=0)=0$ to get
\begin{align*}
  \E[\varepsilon^p_K \phi^2(X)]&=\frac14\E\left[\left(\min\left(|\widehat{H}^{p}_{K/2}| ,|\widehat{H}^{p,\prime}_{K/2}|\right)^2\right)\phi^2(X) \mathds{1}_{\widehat{H}^{p}_{K/2}\widehat{H}^{p,\prime}_{K/2}<0} \mathds{1}_{\widehat{H}^{p}_{K/2} H^p_X<0}\right] \\
  & \quad +\frac14\E\left[\left(\min\left(|\widehat{H}^{p}_{K/2}| ,|\widehat{H}^{p,\prime}_{K/2}|\right)^2\right)\phi^2(X)\mathds{1}_{\widehat{H}^{p}_{K/2}\widehat{H}^{p,\prime}_{K/2}<0} \mathds{1}_{ \widehat{H}^{p,\prime}_{K/2} H^{p}_X<0}\right] \\
  &\le \frac14\left( \E\left[|\widehat{H}^{p}_{K/2}|^2\phi^2(X) \mathds{1}_{\widehat{H}^{p}_{K/2} H^p_X<0}\right]+\E\left[|\widehat{H}^{p,\prime}_{K/2}|^2\phi^2(X) \mathds{1}_{\widehat{H}^{p,\prime}_{K/2} H^{p}_X<0}\right] \right)=\frac 12 \E\left[|\widehat{H}^{p}_{K/2}|^2\phi^2(X) \mathds{1}_{\widehat{H}^{p}_{K/2} H^p_X<0}\right],
\end{align*}
since $\widehat{H}^{p}_{K/2}$ and $\widehat{H}^{p,\prime}_{K/2}$ have the same law given $X$. Now, we use that $|\widehat{H}^{p}_{K/2}|\le|\widehat{H}^{p}_{K/2}- H^p_X|$ on $\{\widehat{H}^{p}_{K/2} H^p_X<0\}$ and Lemma~\ref{lemma::indicatrice_cv_markov} gives $\mathds{1}_{\widehat{H}^{p}_{K/2} H^p_X<0}\le \frac{|\widehat{H}^{p}_{K/2}-H^p_X|^\eta}{|H^P_X|^\eta}$ for $\eta>0$. This leads to
$$ \E[\varepsilon^p_K \phi^2(X)]\le \frac 12 \E\left[\frac{|\widehat{H}^{p}_{K/2}-H^p_X|^{2+\eta}}{|H^P_X|^\eta}\phi^2(X)\right].$$
We now use Lemma~\ref{lem_speedLLN} and get $\E[|\widehat{H}^{p}_{K/2}-H^p_X|^{2+\eta}|X]\le C_{2+\eta}  \frac{D^p_{2+\eta}(X)}{(K/2)^{1+\eta/2}}$, and therefore $$\E[\varepsilon^p_K \phi^2(X)]\le 2^{\eta/2}C_{2+\eta} \E\left[\frac{D_{2+\eta}^p(X)}{|H^p_X|^\eta} \phi^2(X) \right] .$$ 
Using this bound in~\eqref{majo_var}, we get the claim with $C= 2^{\eta/2}C_{2+\eta}\sum_{p=2}^{P} 2^{P+1-p} \E[\frac{D_{2+\eta}^p(X)}{|H^p_X|^\eta}\phi^2(X)] $. 
\end{proof}

  {\bf Acknowledgments. } This research benefited from the Joint Research Initiative ``Numerical methods for the ALM'' of AXA Research Fund. A. A. has also benefited from the support of the ``Chaire Risques Financiers'', Fondation du Risque. We thank Vincent Jarlaud and the team ALM of AXA France for useful discussions and remarks.

\bibliographystyle{plain}
\bibliography{biblio}

\end{document}